\documentclass[sigconf]{acmart}
\renewcommand\footnotetextcopyrightpermission[1]{} % removes footnote with conference information in first column
\usepackage{booktabs} % For formal tables
\usepackage{ifthen}
\newboolean{picturepage}
\setboolean{picturepage}{false}

\usepackage{amsfonts}
\usepackage{xspace}
\usepackage{amsmath, amsthm, amssymb}
\usepackage{subcaption}
\usepackage{graphicx}
\usepackage{url}
\usepackage{hyperref}
\usepackage{algorithm,algpseudocode}
\usepackage{graphicx}
\usepackage{subcaption}
\usepackage{multirow}

\newtheorem{lemma}{Lemma}

\newtheorem{theorem}{Theorem}
\newtheorem{observation}{Observation} 
\theoremstyle{definition}
\newtheorem{definition}{Definition}
\newtheorem{construction}{Construction}

\usepackage{comment}

\newcommand{\y}[1]{{\color{black} #1}\normalcolor}
\newcommand{\jj}[1]{{\color{black} #1}\normalcolor}
\newcommand{\dd}[1]{{\color{black} #1}\normalcolor}
\newcommand{\jjw}[1]{{\color{black} #1}\normalcolor}

\algdef{SE}[VARIABLES]{Variables}{EndVariables}
   {\algorithmicvariables}
   {\algorithmicend\ \algorithmicvariables}
\algnewcommand{\algorithmicvariables}{\textbf{local variables}}

\begin{document}
\title{Distributed Transactional Systems Cannot Be Fast}

\author{Diego Didona$^1$, Panagiota Fatourou$^2$, Rachid Guerraoui$^1$,\\Jingjing Wang$^1$, Willy Zwaenepoel$^{1,3}$\\~\\
$^1$ EPFL; $^2$ FORTH, ICS $\&$ University of Crete, CSD; $^3$ University of Sydney}

\renewcommand{\shortauthors}{D. Didona et al.}

\begin{abstract}
We prove that no fully transactional system can provide fast read transactions
(including read-only ones that are considered the most frequent in practice).
Specifically, to achieve fast read transactions, the system has to give up support of 
transactions that write more than one object.
We prove this impossibility result for distributed storage systems that are causally consistent, 
i.e., they do not require to ensure any strong form of consistency. Therefore, our result 
holds also for any system that ensures a consistency level stronger than causal consistency, 
e.g., strict serializability. 
The impossibility result holds 
even for systems that store only two objects
(and support at least two servers and at least four clients).
It also holds for systems that are partially replicated. 
Our result justifies the design choices of state-of-the-art distributed transactional systems
and insists that system designers should not put more effort 
to design fully-functional systems that support both fast read transactions and
ensure causal or any stronger form of consistency.
\end{abstract}

%\keywords{Causal consistency; distributed storage; fast read-only transactions; impossibility; multi-object transactions}

\maketitle
\pagenumbering{arabic} %Page again
\section{Introduction}
\label{sec:intro}
Transactions represent a fundamental abstraction of distributed storage systems, 
as they facilitate the task of building correct applications. For this reason, distributed transactional storage systems are widely adopted in production environments~\cite{cosmosdb,voltdb,sana,cockroachdb,coherence,mysql,qiao_brewing_2013} and actively researched in academia~\cite{Corbett:2013,Wei:2015,Mu:2014,calvin_thomson_2012,Zhang:2015}.
Because many  applications exhibit read-dominated workloads~\cite{Nishtala:2013,Atikoglu:2012,Bronson:2013,Noghabi:2016}, 
read-only transactions are a particularly important building block 
of such systems.
Hence, improving the performance of distributed read-only transactions
has become a key requirement for modern such systems 
and a much investigated research topic~\cite{Corbett:2013,lu_snow_2016,lloyd_stronger_2013}.
To this end, the notion of {\em fast} read-only transactions 
has been introduced in~\cite{lu_snow_2016} and studied in subsequent papers~\cite{Konwar:2018,tomsic_2018,Didona:2018}.
 A fast read-only transaction 
satisfies the following three desirable properties~\cite{lu_snow_2016}:
it completes in one round of communication ({\em one-round}), 
it \dd{does} not rely on blocking mechanisms ({\em nonblocking}),
and each server communicates to the client only one value for each object that it stores locally and is being read  ({\em one-value}). 
Unfortunately, despite the huge effort put on designing efficient distributed transactional systems,
read-only transactions in existing systems still suffer from performance limitations. 
For example, systems like Spanner~\cite{Corbett:2013}, DrTM~\cite{Wei:2015}, 
RoCoCo~\cite{Mu:2014} implement read-only transactions that may require multiple rounds 
of communication, or rely on blocking mechanisms.

In systems that implement strong consistency, e.g., serializability~\cite{bernstein_concurrency_1987}, 
transactions facilitate the task of writing distributed applications 
by giving the illusion that concurrent operations take place sequentially. 
However, strong consistency comes at such a high cost~\cite{Bailis:2013,osdi}
that, over the last years, a flurry of systems has abandoned strong consistency 
in favor of weaker consistency models~\cite{DeCandia:2007,Cooper_PNUTS_2008,lloyd_settle_2011,qiao_brewing_2013}.
Among them, causal consistency~\cite{ahamad_causal_1995} has garnered much attention, 
because it was expected to hit a sweet spot in the performance versus ease-of-programming trade-off.
Causal consistency has intuitive semantics and
eschews the synchronization that is needed to achieve strong consistency in the presence of replicas. 
It is also the strongest consistency level that can tolerate network partitions 
without blocking operations~\cite{mahajan_consistency_2011,Attiya:2015}.

As expected, existing causally consistent storage systems achieve higher performance in comparison to
strong consistency systems~\cite{lloyd_settle_2011,lloyd_stronger_2013,akkoorath_cure_2016,du_gentlerain_2014}.
However, 
causally consistent read-only transactions still suffer from latency overheads.
In fact, 
state-of-the-art causally consistent storage systems 
either do not support fast read-only transactions~\cite{du_gentlerain_2014,mehdi_occult_2017,akkoorath_cure_2016,Spirovska:2018}
(i.e., they do not exhibit all three desirable properties) or they are of restricted functionality 
by not providing multi-object write transactions~\cite{lu_snow_2016}.

\vspace{-9pt}
~\\\noindent{\bf Contributions.} In this paper, 
we present a result proving a fundamental limitation of transactional systems.
Specifically, our impossibility result states that 
no fully-functional, causally consistent distributed transactional system can provide fast read-only transactions
(and therefore also fast read transactions).  
Specifically, to achieve fast read transactions, the system has to give up support of multi-object
write transactions, i.e., it can only support transactions that write at most one object.
This result unveils an important trade-off between 
the latency attainable by read-only transactions and the functionality provided by a distributed storage system. 
It also shows that the inefficiency of the existing systems to achieve all the desirable properties (as described above)
is not a coincidence.

Most theoretical results considered so far serializable transactions 
instead of causally consistent that we consider here, and
this work includes a formalization of 
causally consistent transactional systems which is interesting in its own right.  
Our result holds for any system that ensures stronger consistency than causal consistency. 
Moreover, our result is relevant for the broad class of systems that use causal consistency as a building block 
to achieve their target consistency level~\cite{balegas_indigo_2015} or that implement hybrid consistency models 
which include causal consistency~\cite{osdi,Li:2014,Ardekani:2014,Terry:2013}. 
The impossibility result holds for any system that 
supports at least two servers and at least four clients.
It  holds even for systems that store only two objects
(each in a different server).  
We prove that the impossibility result also holds for systems that are partially replicated. 

To prove our impossibility result, we construct a troublesome infinite execution
in which a write-only transaction that is executed solo never manages 
to make the values it writes visible.
To do so, we %have to 
 inductively construct an infinite number of non-empty prefixes 
of the troublesome execution and prove that the written values are
not yet visible after each prefix has been executed; specifically, some server
has to send at least one more message before the values become visible. 
We argue this using indistinguishability arguments~\cite{Attiya:2004:DCF:983102,Lynch:1996:DA:525656}. 
The fact that, on the one hand, the constructed execution 
is infinite, and, on the other, that we focus on causal consistency which is
a rather weak consistency  condition, introduces complications 
that we have to cope with to get the proof. 

In the case where %we consider
the system has more than two servers, an extra challenge is %that we have 
to cope with the chains of messages through which information may be disseminated
between the servers that store the written objects. This complicates the
construction of the executions that we prove to be indistinguishable 
from the troublesome execution. 
To get the impossibility result for the case of a partially replicated system, an additional
complication is that we have to construct an infinite sequence of server ids.  
These are the servers that send the necessary messages in each step 
for the induction to work. We also have to %adjust some of our definitions to 
capture the fact that more than one servers may now respond to the same
read request of a client. Due to lack of space, the general proof %for the general case
is provided in an attached appendix.

We study the limits of our impossibility result in all different premises. 
We show that  
if we relax any of the considered properties, 
then the impossibility result no longer holds.

Our impossibility result sheds light on some of the design decisions of recent systems
and provides useful information to system designers. Specifically, they should not put 
more effort to devise a fully-functional system that supports both fast read transactions and
ensures causal consistency (or any stronger consistency level). 

\noindent{\bf Structure of the paper.} Section 2 provides the model and useful definitions. 
Section 3 presents our impossibility result and discusses its limits in all different premises. 
Section 4 discusses related work. Section 5 provides some conclusions.

\section{Model and Definitions}
\label{sec:model-and-def}
\noindent{\bf Storage system.} We consider a distributed storage system which stores a finite number of objects. 
There are $m > 1$ servers in the system. Each server stores a non-empty set
of these objects.
For simplicity, we assume that the set of objects stored in servers 
are disjoint.
(Our result holds even if the system 
is {\em partially replicated}, i.e., if these sets are different but not disjoint,
and none of them contains all the objects.)

\noindent{\bf Transactions.} An arbitrarily large number of clients may read and/or write one or more objects 
by  executing \emph{transactions}. 
To prove our impossibility results, it is enough to focus on {\em static} transactions whose read-sets and write-sets are known 
from the beginning of the execution. (It follows that our impossibility result holds for systems of dynamic transactions as well.) 
A (static) transaction $T = (R_T, W_T)$ reads the objects in its {\em read-set}, $R_T$, 
and writes the objects in its {\em write-set}, $W_T$. 
If $W_T = \emptyset$, $T$ is called a {\em read-only} transaction, whereas
if $R_T = \emptyset$, $T$ is a {\em write-only} transaction.  
We denote by $r(X)x$ a read on object $X$ which returns {\em value} $x$
and by $w(X)x$ a write of {\em value} $x$ to object $X$. 
Also, we denote by $r(X)*$ a read on object $X$ when the return value is unknown (with symbol $*$ as a place-holder).
%Each object is initialized with a special value $\bot$.
Reads and writes on objects
are called {\em object operations}.

In typical deployments, there are many more clients than servers. Hence, 
allowing server-to-client out-of-band communication  would result in nonnegligible overhead 
on the servers,  which would suffer from reduced  scalability and performance.
For this reason, we naturally assume that to execute a transaction, 
a client can communicate with the servers (but not with other clients)
and a server communicates with a client only to respond to a client's read or write {\em request}. 
We find evidence of the relevance of this assumption in large-scale production systems, 
such as Facebook's data platform~\cite{Nishtala:2013}, 
and in emerging systems~\cite{Lim:2014,Jin:2017,Didona:2019} and architectures~\cite{Li:2015} 
for fast query processing,  where 
no per-client states are maintained to avoid the corresponding overheads and to achieve the lowest latency.

\noindent{\bf System model.} The system is \emph{asynchronous}, i.e., the delay on message transmission 
can be arbitrarily large and the processes do not have access to a global clock. 
The system is modelled as an undirected graph in the standard way~\cite{Attiya:2004:DCF:983102,Lynch:1996:DA:525656}.
Each node of the graph represents a process (i.e., a client or a server)
whereas links connect every pair of processes. Each process is modelled as a state machine
with its state containing a set of income and outcome buffers~\cite[Ch. 2]{Attiya:2004:DCF:983102}\footnote{
There is one income and one outcome buffer for each link incident to each process. 
Income and outcome buffers store the messages that are sent or received through the corresponding link, respectively. }. 
Links do not lose, modify, inject, or duplicate messages.

\noindent{\bf Operation executions.} An {\em implementation} of a storage system provides algorithms, for each process, 
to execute reads and writes in the context of transactions. 
A {\em configuration} represents an instance of the system at some point in time. 
In an {\em initial} configuration 
$Q_{in}$,  
all processes are in initial states and all buffers are empty (i.e., 
no message is in transit). There are two kinds of {\em events} in the system: 
(1) a {\em computation step} taken by a process, in which the process
reads all messages residing in its income buffers, performs some local computation and may send 
(at most) one message to each of its neighboring processes, and (2) a {\em delivery} event, where a message
is removed from the outcome buffer of the source and is placed in the income buffer of the destination. 
An {\em execution} is a sequence of events (we assume that an execution also includes 
the invocations and responses of transactions, as well as the invocations and responses of object operations). 
An execution $\alpha$ is {\em legal} starting from a configuration $C$, 
if, for every process $p$, every computation step taken by $p$ in $\alpha$ 
is compatible to $p$'s state machine (given $p$'s state in $C$)
and all messages sent are eventually received.
Since the system is asynchronous,
the order in which the events appear in an execution 
is assumed to be controlled by an {\em adversary}.
A {\em reachable} configuration is a configuration that results from the application of a legal finite execution
starting from an initial configuration. % in which all buffers are empty (i.e., no message is in transit).} %$C_0$.
Given a reachable configuration $C$, we say that a computation step $s$ by a process $p$ {\em eventually occurs}, 
if in every legal execution starting from $C$ (in which $p$ takes a sufficient number of steps), $p$ 
executes $s$. 
For every reachable configuration $C$ and every legal execution $\alpha$ from $C$,
we denote by $RC(C,\alpha)$ the configuration that results from the execution
of $\alpha$ starting from $C$.
Two executions are {\em indistinguishable} to a process $p$, 
if $p$ executes the same steps in each of them. 
Two configurations are {\em indistinguishable} to $p$, if $p$ is in the same state 
in both configurations.
Given two executions $\alpha_1$ and $\alpha_2$, we denote by $\alpha_1 \cdot \alpha_2$
the concatenation of $\alpha_1$ with $\alpha_2$, i.e., $\alpha_1 \cdot \alpha_2$ is an execution
consisting of all events of $\alpha_1$ followed by all events of $\alpha_2$ (in order).

Each client that {\em invokes} a transaction $T$ may eventually return a {\em response}. 
The response consists of a value for each object in $R_T$, 
and an $ack$ for each write $T$ performs. 
%Then, 
We say that $T$ has {\em completed} (or is {\em complete}), if the client $c$ that invoked $T$ 
has issued all read or write requests for $T$, and has received responses by the servers for all
these requests. 
A transaction $T$ is {\em active} in some configuration $C$
if it has been invoked before $C$ and has not yet completed. 
We say that a configuration $C$ is {\em quiescent} if no transaction is active at $C$. 
Let $C$ be any reachable configuration that is either quiescent or only a single transaction $T$, 
invoked by a client $c$, is active at $C$. 
We say that \emph{$T$ executes solo} starting from $C$, if only $c$ and the servers take steps
after $C$ and $c$ does not invoke any transaction other than $T$ after $C$.
We say that a value $x$ is {\em written} in an execution $\alpha$
if there exists some transaction $T$ in $\alpha$ that issues $w(X)x$ for some object $X$.
For convenience, every execution we consider, starts with 
the execution of two initial transactions, $T_0^{in} = (w(X_0)x_0^{in})$ (invoked by a client $c_0^{in}$) 
and $T_1^{in} = (w(X_1)x_1^{in})$
(invoked by a client $c_1^{in}$) that write the initial values $x_0^{in}$ and $x_1^{in}$ 
in objects $X_0$ and $X_1$, respectively.

%\subsection{Causal Consistency}
\noindent
{\bf Causal Consistency.}
We consider an implementation of a transactional storage system 
that is \emph{causally consistent}~\cite{ahamad_causal_1995, raynal_from_1997}.  
Informally, causal consistency ensures that 
causally-related transactions should appear, to all processes, 
like if they have been executed in the same order.
The formal definitions below closely follow those presented for causally 
consistent transactional memory systems in~\cite{Dziuma2015}.

The {\em history} $H(\alpha)$ of an execution $\alpha$ is the subsequence of $\alpha$ 
which contains only the invocations and responses of object operations (we 
omit $\alpha$ whenever it is clear from the context).
A transaction $T$ {\em is} in $H$, if $H$ contains at least 
one invocation of an object operation by $T$.  
%An operation is {\em complete} in $H$, 
%if there is a response for it in it; otherwise, it is {\em active}.
A transaction is complete in $H$ if it is complete in $\alpha$. 
%A transaction is {\em complete} in $H$, 
%if there is a response for every operation it issues; 
%otherwise, it is {\em active}.
For each client $c$, 
%we denote by $H_{c,w}$ the subsequence of $H$ containing all invocations
%and responses of object operations issued or received by $c$ and 
%all invocations and responses of write operations issued by transactions
%invoked by clients other than $c$ in $H$;
we denote by $H_c$ the subsequence of $H$ containing all invocations
and responses of object operations issued or received by $c$.
Given two executions $\alpha$ and $\alpha'$, 
the histories $H(\alpha)$ and $H(\alpha')$ are
{\em equivalent}, if for each client $c$, $H_c(\alpha) = H_c(\alpha')$. 
For each client $c$, the {\em program-order}, denoted by $<_{H|c}$, 
is a relation on transactions in $H_c$ 
such that for any two transactions $T_1, T_2$, 
$T_1 <_{H|c} T_2$ if and only if $T_1$ precedes $T_2$ in $H_c$.

We denote by $complete(H)$ the subsequence of all events in $H$ issued and received 
by complete transactions\footnote{We assume that in a system that 
supports causal or any weaker form of consistency, all transactions commit~\cite{Bailis:2014,lloyd_stronger_2013}.}. 
\jj{We denote by $comm(H)$ a history that extends $H$: (1) $comm(H) = H \cdot H''$, i.e., $H$ is a prefix of $comm(H)$, (2) $H''$ contains only the responses for those write operations 
in $H$ for which there is no response.
}

A transaction $T_1$ {\em precedes} another transaction $T_2$ 
in $\alpha$ (or in $H(\alpha)$), if $T_1$ completes 
before $T_2$ is invoked. 
Similarly, an object operation $op_1$ {\em precedes} another object operation $op_2$
in $\alpha$ (or in $H(\alpha)$), if the response of $op_1$ precedes the invocation
of $op_2$. 
An execution $\sigma$ is {\em sequential} if for every two  
transactions $T_1$, $T_2$ in $\sigma$, either $T_1$ precedes $T_2$ or vice versa. 
We define a {\em sequential history} in a similar way.

Consider a sequential execution $\sigma$ (legal from %$C_{in}$) 
\jj{$Q_{in}$)}
and a transaction $T$ in $\sigma$. 
We say that $T$ is {\em legal} in $\sigma$,
if for every invocation $r(X)v$ of a read operation on any object $X$ 
%(whose response is $v$) 
that $T$ performs, the following hold:
(1) if there is an invocation of a write operation by $T$ 
that precedes $r(X)v$ in $\sigma$ then $v$ is the value argument of the last such invocation,
(2) otherwise, if there are no transactions preceding $T$ in
$\sigma$ which invoke write for object $X$, then $v$ is the initial value for $X$,
(3) otherwise, $v$ is the value argument of the last invocation of a write operation 
to $X$, by any transaction that precedes $T$ in $\sigma$.

%Consider any o-sequential history $S$ that is equivalent to $H$. 
\jj{Consider any sequential history $S$ that is equivalent to $H$. }
We define a binary relation with respect to $S$, called {\em reads-from} and denoted by $<^r_{S}$,
on {\em transactions} in $H$ such that, for any two distinct transactions $T_1, T_2$ in $H$, 
$T_1 <^r_{S} T_2$ if and only if: (1) $T_2$ executes a read operation $op$ 
that reads some object $X$ and returns a value $v$ for it, and
(2) $T_1$ is the transaction in $S$ which executes the last write operation that writes $v$ for $X$
\jj{and precedes $T_2$.}
%and precedes $op$.
%\item no other transaction $T'$ is related with $T_2$ through by $<^r_H$, i.e. for each pair in 
%\item[(b)] $\langle T', T \rangle$ does not create a cycle in $<^r_H$.
%\end{itemize}
%
%Notice that each o-sequential history $S$ that is equivalent to $H$, 
Each sequential history $S$ that is equivalent to $H$, 
induces a {\em reads-from} relation for $H$. %Aparently, some of these relations may be different to each other.
%Notice that each of these relations contains only a single pair $(T_1,T_2)$ where $T'$ can be any of these transactions containing such an instance of \WriteDI. 
Let ${\bf R}_H$ be the set of all reads-from relations that can be induced for $H$
(by considering all equivalent to $H$ sequential histories).
 
We say that a sequential history $S$ respects some relation $<$ 
%on the set of object operations in $H$
\jj{on the set of transactions in $H$}
if it holds that 
%for any two object operations $op_1$ and $op_2$ in $S$, 
\jj{for any two transactions $T_1$, $T_2$ in $S$,}
%if $op_1 < op_2$, then $op_1$ precedes $op_2$ in $H$.
\jj{if $T_1 < T_2$, then $T_1$ precedes $T_2$ in $S$.}

For each $<^r$ in ${\bf R}_H$, we define the {\em causal} relation for $<^r$ on transactions in $H$ 
to be the transitive closure of  
%$\bigcup_{i} \left(<_{H|p_i}\right) \cup <^r$.
\jj{$\bigcup_{c} \left(<_{H|c}\right) \cup <^r$.}
%We remark that the causal relation for $<^r_H$ is not necessarily a partial order since it may contain cycles.
%
We denote by ${\bf C}_H$ the set of all causal relations in $H$.
\vspace{-5pt}
\begin{definition}[Causal consistency]
An execution $\alpha$ is {\em causally consistent} 
%if there exists a history $H' \in comm(H(\alpha))$ and
if for history $H' = comm(H(\alpha))$, there exists
%, for each prefix $H_p$ of $H$, 
a causal relation $<^c$ in
 ${\bf C}_{complete(H')}$ such that, for each client $c_i$, there exists
% a history $H_i' \in Complete(H)$ %($H_i' \in Complete_2(H_p)$, respectively) and 
a sequential execution $\sigma_i$ such that:
\begin{itemize}
  \item  $H(\sigma_i)$ is equivalent to $complete(H')$, 
  \item $H(\sigma_i)$ respects the causality order $<^c$, and
  \item every transaction executed by $c_i$ in $H(\sigma_i)$ is legal.
\end{itemize}
An implementation is {\em causally consistent} if each execution $\alpha$ it produces
is causally consistent.
\end{definition}
\vspace{-5pt}
Intuitively, a causally consistent distributed transactional system 
produces executions that respect the causality order. 
For simplicity, assume that all values written in an execution $\alpha$ 
are distinct. Then, if a client $c$ reads $x_0$  and $x_1$ for two objects $X_0$
and $X_1$, and $x_0 <^c x_1$, then there is no $x'_0$ such that 
$x_0 <^c x'_0 <^c x_1$.
The necessity to talk about sets of reads-from relations 
and sets of causal relations above comes from the fact
that the values written in an execution $\alpha$ are not necessarily distinct.

%\subsection{Progress}
\noindent
{\bf Progress.}
To avoid trivial implementations in which every read-only transaction invoked by any client 
always returns $\bot$ or values written by the same client, we introduce the concept of value visibility. 
\vspace{-5pt}
\begin{definition}[Value visibility]
\label{def:vis-val}
Consider any object $X$ and let $C$ be any reachable configuration
which is either quiescent or just a write-only transaction
(by a client $c_w$)
	writing a value $x$ into $X$ (and possibly performing additional writes) is active in $C$. 

Value $x$ is {\em visible} in $C$, if and only if: in every legal execution 
starting from $C$ which contains just a read-only transaction $T_r$ 
(invoked by a client  $c \not\in \{c_w, c_0^{in}, c_1^{in} \}$)
that reads $X$, $x$ is returned as $X$'s value for $T_r$.
\end{definition}
\vspace{-5pt}
We focus on storage systems that ensure minimal progress for write-only transactions.
This is a weak progress property ensured even by systems 
in which write transactions are blocking.  So, our impossibility result  
holds also for systems that ensure any stronger progress properties for write-only transactions
(and without any restriction on progress for transactions that both read and write).

\vspace{-5pt}
\begin{definition}[Minimal Progress for write-only transactions]
\label{def:prog}
Let $T_w$ be any write-only transaction which writes a value $x$ into an object $X$ ($T_w$ may also write other objects). 
If $T_w$ executes solo, starting from any reachable quiescent configuration $C$,
then there exists a later configuration $C'$ such that $x$ is visible in $C'$.
\end{definition}
\vspace{-5pt}

We denote by \jj{$Q_0$} a reachable configuration in which both values
$x_0^{in}$ and $x_1^{in}$, written by the transactions $T_0^{in}$ and $T_1^{in}$ discussed earlier, 
are visible. Definition~\ref{def:prog} implies that $Q_0$ is well-defined.

We next present the definition of fast read-only transactions.
Definition~\ref{def:fast-op} expresses the exact same properties as in the original definition 
which was introduced in~\cite{lu_snow_2016} and used in~\cite{Didona:2018,tomsic_2018,Konwar:2018}. 
\vspace{-5pt}
\begin{definition}[Fast read-only transaction]
\label{def:fast-op}
We say that an implementation of a distributed storage system supports 
{\em fast read-only transactions}, if in each execution $\alpha$ it produces, 
the following hold for every read-only transaction $T$ executed in $\alpha$: 
\begin{enumerate}
\item
{\bf Non-blocking and One-Roundtrip Property.} 
The client $c$ which invoked $T$ 
sends a message to all servers storing items that it wants to read
and it does so in one computational step; moreover,
each server performs at most one computational step to serve the request
and respond to $c$.
\item
{\bf One-value messages.}
 Each message sent from a server to a client contains only one value 
that has been written by some write transaction in $\alpha$ into an object 
that is stored in the server and is read by the client\footnote{We remark that the message may also contain
some metadata (e.g., a timestamp), as long as these metadata 
do not reveal any information about other transactions and additional written values to the client.}.
\end{enumerate}
\end{definition}
\vspace{-5pt}

\section{The Impossibility Result}
\label{sec:imp}

In this section, we prove the impossibility result: 

\vspace{-2pt}
\begin{theorem}
\label{thm:imp}
No causally consistent implementation of a transactional storage system 
supports transactions that write to multiple objects and 
fast read-only transactions.
\end{theorem}
\vspace{-2pt}

The impossibility result holds even for systems that store just two objects $X_0$ and $X_1$. 
For ease of presentation, we prove 
it for a system with two servers $p_0$ and $p_1$.
However, it can
easily be extended to hold for the general case, where the system
has any number of servers and is partially-replicated (see the appendix
for the general proof).
We assume that $p_0$ stores $X_0$ and $p_1$ stores $X_1$.

\subsection{Outline of the Proof}
We prove Theorem~\ref{thm:imp} by the way of contradiction. Table~\ref{table:symbols} in Appendix reports the list of the main symbols used throughout the proof.
Assume that there exists a causally consistent implementation $\Pi$ 
which supports fast read-only transactions and %concurrent 
transactions that write multiple objects.
Assume also that $\Pi$
guarantees minimal progress for write-only transactions. 
We derive a contradiction by showing that there exists a {\em troublesome} execution which breaks minimal progress. 
Specifically, we construct an infinite execution $\alpha$ which contains 
just a write-only transaction $T_w = (w(X_0)x_0, w(X_1)x_1)$, 
invoked by some client $c_w\notin \{c_0^{in}, c_1^{in}\}$,
and we show that the values $x_0$ and $x_1$ written by $T_w$ never become visible. 
%(The {\em troublesome} execution in Section \ref{sec:intro} refers to this execution $\alpha$.)
Intuitively, to do so, we inductively construct an infinite number of non-empty distinct prefixes 
of $\alpha$ and prove that the written values are
not yet visible after each prefix has been executed.
Specifically, we use indistinguishability arguments~\cite{Attiya:2004:DCF:983102,Lynch:1996:DA:525656}
to prove that after the execution of each such prefix, 
some server has to still send at least one message before the values become visible. 

\y{We now provide an informal, high-level outline of the proof.
%description of how we get to build the troublesome execution to prove the theorem.  
We start with two simple lemmas.
The first shows that a transaction which reads $X_0$ and $X_1$
cannot return a subset of the values written by $T_w$ (i.e.,
it returns either the new values for both objects or the initial values for both objects). 
The second lemma shows that if one of the values written by $T_w$ is not visible,
then both values written by $T_w$ are not visible. 
We use these lemmas to determine the values read 
by read-only transactions in the executions that  we use. 
Specifically, we present two execution constructions that are useful 
in the proof of Theorem~\ref{thm:imp}. 
Construction~\ref{gold_construction} describes an execution in which a read-only transaction 
reads the initial values for $X_0$ and $X_1$, whereas
Construction~\ref{gnew_construction} describes 
an execution in which a read-only transaction reads the new values for $X_0$ and $X_1$.  
     
The two constructions are used to build an execution $\gamma$
which allows us to derive a contradiction: 
In $\gamma$, a read-only transaction reads 
a mix of old and new values for $X_0$ and $X_1$. 
We use $\gamma$ to prove, in the induction, that
the values written by $T_w$ are not yet visible. We also prove that
to make them visible, $p_0$ and $p_1$ have to exchange more messages.
Therefore, after any arbitrary but finite number of steps,
$T_w$ has not yet completed 
and the values written by it has not yet
become visible, which contradicts eventual visibility.}

\subsection{Useful Constructions and Lemmas}

To enforce a causal relation between the values written by $c_w$ in $T_w$
and $x_0^{in}$, $x_1^{in}$, 
\y{the troublesome execution} starts with 
the execution of a read-only transaction $T_r^{in} = (r(X_0)*, r(X_1)*)$ 
by $c_w$ (applied from $Q_0$). 
Since $T_r^{in}$ is a fast read-only transaction, 
$T_r^{in}$ completes; since $x_0^{in}$ and $x_1^{in}$ are visible in $Q_0$
\y{(by definition)}, 
$c_w$ returns $(x_0^{in}, x_1^{in})$ for $T_r^{in}$.
\y{Let $C_0$ be the configuration }
in which $T_r^{in}$ has completed and no message is in transit.
The configurations $Q_{in}$, $Q_0$, $C_0$ are shown in Figure~\ref{fig:init}.
\y{All executions we refer to below start from $C_0$.}
\ifthenelse{\boolean{picturepage}}{}{
\begin{figure}[t]
\vspace{-10pt}
    \centering
    \includegraphics[width=0.4\textwidth]{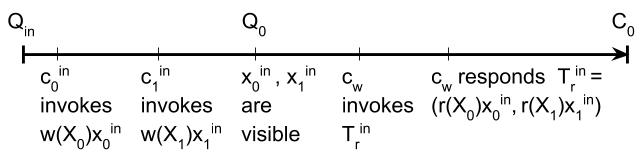}
\vspace{-10pt}
    \caption{Configurations $Q_{in}$, $Q_0$ and $C_0$.}
    \label{fig:init}
\vspace{-10pt}
\end{figure}
}
\y{We now present the two lemmas that will be useful for the constructions 
and the proof of Theorem~\ref{thm:imp}. 
The first states that in an execution in which a write transaction writes new values to a set of objects, }
a read transaction which reads these objects, cannot see only a subset of the new values. 
The proof comes as an immediate consequence of the fact that $\Pi$ ensures causal consistency. 
%\here{Jingjing, Diego: I want an intutive description of the proof of lemma 1 here.}
%
\vspace*{-5pt}
\begin{lemma}
\label{lem:granularity}
Let $\tau$ be any legal execution of $\Pi$ starting from $C_0$ which contains two transactions: 
client $c_w$ invokes a write-only transaction
$T_w = (w(X_0)x_0, w(X_1)x_1)$, and a client $c_r \neq c_w$ invokes a read-only transaction $T_r = (r(X_0)*, r(X_1)*)$ which completes in $\tau$.
Let $v_0$ and $v_1$ be the two values which $c_r$ returns for $T_r$, i.e., $T_r = (r(X_0)v_0, r(X_1)v_1)$.
Then, either $v_0 = x_0$ and $v_1= x_1$, or $v_0 = x_0^{in}$ and $v_1= x_1^{in}$.
\end{lemma}
\vspace*{-5pt}
\begin{proof}[Proof of Lemma \ref{lem:granularity}]
To derive a contradiction, assume that there exists some execution $\tau$ 
of $\Pi$ starting from $C_0$ and some $i\in\{0,1\}$ for which $c_r$ returns
the values
$v_i = x_i, v_{1-i} = x_{1-i}^{in}$  for $T_r$.

By causal consistency, $c_r$'s local history in $\tau$ respects causality. 
So, we can totally order 
$T_{1-i}^{in}$, $T_w$ and $T_r$ such that 
(1) $T_{1-i}^{in}$ is the last transaction which writes to $X_{1-i}$ and precedes $T_r$,
so $T_w$ (which also writes $X_{1-i}$) precedes $T_{1-i}^{in}$;
(2) $T_w$ is the last transaction which writes on $X_{i}$ and precedes $T_r$, so $T_w$ is ordered before $T_r$;
%(and $T_{1-i}^{in}$); %and also before $T_{1-i}^{in}$;
%$T_w$ is also ordered before $T_{1-i}^{in}$ because $T_w$ writes $X_{1-i}$ but $T_{1-i}^{in}$ is the last transaction that writes on $X_{1-i}$ (by the previous bullet);
and (3) $T_{1-i}^{in}$ is ordered before $T_w$ because $c_w$ reads the value written in $T^{in}_{1-i}$ before it initiates $T_w$ (as shown in Figure~\ref{fig:init}).
A contradiction. 
\end{proof}
\vspace*{-5pt}

\y{The next lemma states that in an execution 
in which a write transaction $T_w$ writes new values in a set of objects,
if one of the values written by $T_w$ is not visible at some configuration,
then both values written by $T_w$ are not visible in that configuration. 
}
\vspace*{-5pt}
\begin{lemma}
\label{lem:not-visible}
Let $\tau$ be any legal execution starting from $C_0$ which contains just one transaction: 
client $c_w$ executes a write-only transaction $T_w = (w(X_0)x_0, w(X_1)x_1)$. 
Let $C$ be any reachable configuration when $\tau$ is applied from $C_0$.
If either $x_0$ or $x_1$ is not visible in $C$, then
there exists at least one client
$c_r \not\in \{c_w, c_0^{in}, c_1^{in} \}$
such that if, starting from $C$, $c_r$ executes  a fast read-only transaction 
$T_r = (r(X_0)*, r(X_1)*)$, then $c_r$ returns 
$(x_0^{in}, x_1^{in})$ 
for $T_r$.
\end{lemma}
\vspace*{-5pt}
\begin{proof}[Proof of Lemma \ref{lem:not-visible}]
\y{To derive a contradiction, 
consider some configuration $C$, reachable when $\tau$ is applied from $C_0$, 
in which at least one of the values $x_0$ and $x_1$ is not visible, 
and assume that for every client}
$c_r \not\in \{c_w, c_0^{in}, c_1^{in} \}$,
if $c_r$ invokes $T_r$ starting from $C$, then $c_r$ does not return 
$(x_0^{in}, x_1^{in})$ for $T_r$. 
Notice that since $\Pi$ ensures that read-only transactions are fast, 
$T_r$ completes.
Then by Lemma \ref{lem:granularity}, $c_r$ returns $(x_0, x_1)$. 
Since this holds for every client $c_r$ not in $\{c_w, c_0^{in}, c_1^{in} \}$,
Definition \ref{def:vis-val} 
implies that at $C$ both $x_0$ and $x_1$ are visible. This contradicts the hypothesis that 
at least one of the two values is not visible in $C$.
\end{proof}

\vspace*{-5pt}
Notice that Lemma~\ref{lem:not-visible} holds for $C = C_0$, 
i.e., when $\tau$ is empty. 
%then Lemma~\ref{lem:not-visible} \jjw{also} holds for all clients, 
%not in $\{c_w, c_0^{in}, c_1^{in} \}$.

\y{We now present Constructions~\ref{gold_construction} and~\ref{gnew_construction}
which are illustrated in Figure \ref{fig:cons}. 
Roughly speaking, the constructions illustrate two executions
in which a write-only transaction $T_w$ writes values $x_0$ and $x_1$ to objects $X_0$ and $X_1$, 
respectively, and a read-only transaction $T_r$ reads $X_0$ and $X_1$. 
The executions are constructed so that $T_r$ returns $(x_0^{in}, x_1^{in})$ in the first execution,
whereas it returns $(x_0, x_1)$ in the second.
Based on these constructions, %and their accompanying observations 
%(\ref{obs:gold} and \ref{obs:gnew}),
we construct, in the proof of Theorem~\ref{thm:imp}, an execution 
where $T_r$ returns a mix of initial and new values, allowing us to derive a contradiction.
The constructions are based on 
a fixed $i \in \{ 0, 1\}$ and a client $c_r$ that executes transaction $T_r$. Each time the
construction is employed in the proof of Theorem~\ref{thm:imp} these parameters can be different. 
Although these executions are similar, we present both of them for ease of presentation.}

\ifthenelse{\boolean{picturepage}}{}{
\begin{figure*}
    \centering
  \begin{subfigure}[b]{0.45\textwidth}
    \includegraphics[width=\textwidth]{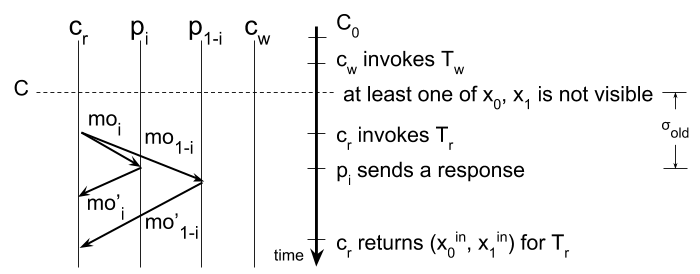}
    \caption{Construction \ref{gold_construction}}
    \label{fig:gold}
    \end{subfigure}\hfill
  \begin{subfigure}[b]{0.45\textwidth}
    \includegraphics[width=\textwidth]{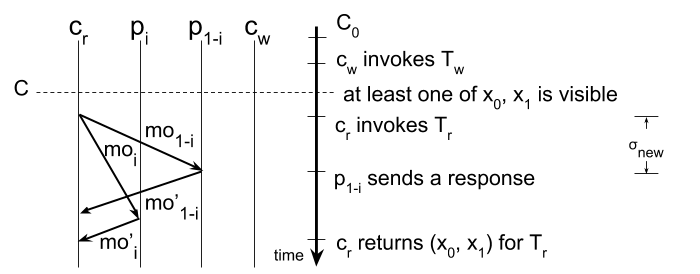}
    \caption{Construction \ref{gnew_construction}}
    \label{fig:gnew}
    \end{subfigure}
\vspace*{-.2cm}
    \caption{Illustration of Constructions~\ref{gold_construction} and~\ref{gnew_construction}.} 
    \label{fig:cons}
\vspace*{-.4cm}
\end{figure*}
}

We start with an intuitive description of Construction~\ref{gold_construction} which is depicted in Figure 2(a). 
Assume, without loss of generality, that $i  = 0$.
(The construction when $i = 1$ is symmetric.)
\y{The construction produces an execution in which $T_w$
starts executing from $C_0$ and runs solo up to any configuration $C$ in which $x_0$ 
is not yet visible. Next, $T_r$ is initiated from $C$ and takes steps 
until it sends a message to both servers. The adversary 
schedules the receipt of these messages so that $p_0$ receives the message first
and sends a response. 
Then, $p_1$ receives the message sent by $c_r$ 
and sends back a response. Finally, $c_r$ takes steps to collect these responses
and return. We call $\sigma_{old}$ the part of the execution starting from $C$
until the point that $p_0$ sends a response, and $\gamma_{old}$ the suffix of the execution 
starting from $C$. \y{Lemma \ref{lem:not-visible} allows us to argue that} 
$c_r$ returns $(x_0^{in}, x_1^{in})$ in $\gamma_{old}$.
We next present the formalism of the construction. }
\vspace*{-.1cm}
\begin{construction}[Construction of execution $\gamma_{old}(C, p_i, c_r)$ and execution $\sigma_{old}(C, p_i, c_r)$]
\label{gold_construction}
Let $\tau$ be any arbitrary legal execution starting from $C_0$ which contains just one transaction: 
client $c_w$ executes a write-only transaction $T_w = (w(X_0)x_0, w(X_1)x_1)$. 
Fix any $i \in \{ 0, 1\}$. 
\y{For every configuration $C$ that is reached when $\tau$ is applied from $C_0$
in which {\em the value $x_i$ is not visible},
Lemma~\ref{lem:not-visible} implies that 
there exists at least one client
$c_r \not\in \{c_w, c_0^{in}, c_1^{in} \}$
such that if, starting from $C$, $c_r$ executes  a fast read-only transaction 
$T_r = (r(X_0)*, r(X_1)*)$, then $c_r$ returns 
$(x_0^{in}, x_1^{in})$ 
for $T_r$.}
For every  such client $c_r$, %that satisfies Lemma~\ref{lem:not-visible},
we define $\gamma_{old}(C, p_i, c_r)$ to be the execution containing all of the events described below. 
In $\gamma_{old}(C, p_i, c_r)$, first $c_r$ invokes $T_r$ starting from $C$. 
So, $c_r$ takes steps and since it reads $X_0$ and $X_1$,  
$c_r$ sends a message $mo_0(C, p_i, c_r)$ to $p_0$ and a message $mo_1(C, p_i, c_r)$ to $p_1$.
Next, the adversary schedules the delivery of $mo_i(C, p_i, c_r)$ 
and let $p_i$ take a step and receive $mo_i(C, p_i, c_r)$.
Since $T_r$ is a fast transaction, once $p_i$ receives $mo_i(C, p_i, c_r)$, 
$p_i$ sends a response $mo_i'(C, p_i, c_r)$ to $c_r$. % in its next step.
\y{Denote by $\sigma_{old}(C, p_i, c_r)$ the prefix of $\gamma_{old}(C, p_i, c_r)$ %sequence of steps starting from $C$ 
up until the step in which $p_i$ sends the response.}
After $\sigma_{old}(C, p_i, c_r)$ has been applied from $C$,
the adversary schedules the delivery of $mo_{1-i}(C, p_i, c_r)$, and lets 
$p_{1-i}$ take  steps to receive $mo_{1-i}(C, p_i, c_r)$
and send a response to $c_r$. % in its next step.
Finally, $c_r$ take steps to
receive the responses  from $p_0$ and $p_1$
and return a response for $T_r$. 
\end{construction}
\vspace*{-.1cm}
By the way $\gamma_{old}(C, p_i, c_r)$ is constructed and by 
Lemma \ref{lem:not-visible}, we get the following.
\vspace*{-.1cm}
\begin{observation}
\label{obs:gold}
The following claims hold:
\begin{enumerate}
\item \label{gold_legal} Execution $\gamma_{old}(C, p_i, c_r)$ is legal from $C$. 
\y{\item \label{gold_ind} Only processes $p_i$ and $c_r$ take steps in $\sigma_{old}(C, p_i, c_r)$,
so configurations $C$ and $RC(C, \sigma_{old}(C, p_i, c_r))$ are indistinguishable 
to $c_w$ and $p_{1-i}$}.
\item \label{gold_return} The return value for $T_r$ in $\gamma_{old}(C, p_i, c_r)$ is  
($x_0^{in}, x_1^{in}$).
\end{enumerate}
\end{observation}
\vspace*{-.1cm}
\y{Construction~\ref{gnew_construction} is depicted in Figure 2(b). 
Again, assume, without loss of generality, that $i  = 0$.
(The construction for the case where $i = 1$ is symmetric.)
The construction produces an execution in which $T_w$
starts its execution from $C_0$ and runs solo up to any configuration $C$ in which $x_1$ 
is  visible. Then, $T_r$ is initiated from $C$ and $c_r$ takes steps 
until it sends a message to each of the servers. 
The adversary schedules the delivery of these messages 
so that $p_1$ receives the message first and sends back a response. 
Then, $p_0$ receives the message sent by $c_r$ 
and sends back a response. Finally, $c_r$ takes steps to collect these responses
and return. We call $\sigma_{new}$ the part of this construction starting from $C$
until the point that $p_1$ sends a response, and $\gamma_{new}$ the suffix of this execution 
starting from $C$.  We later argue (in Observation~\ref{obs:gnew}) 
that $c_r$ returns $(x_0, x_1)$ in $\gamma_{old}$.
We next present the formalism of the construction.  
\vspace*{-.1cm}
\begin{construction}[Construction of execution $\gamma_{new}(C, p_i, c_r)$ and execution $\sigma_{new}(C, p_i, c_r)$]
\label{gnew_construction}
Let $\tau$ be any legal execution starting from $C_0$ which contains just one transaction: 
client $c_w$ executes a write-only transaction $T_w = (w(X_0)x_0, w(X_1)x_1)$.
Fix any $i \in \{ 0, 1\}$ and let $c_r$ be any client 
not in $\{c_w, c_0^{in}, c_1^{in}\}$.
For every reachable configuration $C$ when $\tau$ is applied from $C_0$
in which 
{\em $x_i$ is visible},
we define $\gamma_{new}(C, p_i, c_r)$ to be the execution containing all of the events described below.
In $\gamma_{new}$, $c_r$ invokes $T_r$ starting from $C$ and takes steps until  
it sends a message $mn_0(C, p_i, c_r)$ to $p_0$ and a message $mn_1(C, p_i, c_r)$ to $p_1$.
Let $C_{new}(C, p_i, c_r)$ be the resulting configuration.
Next, the adversary schedules the delivery of $mn_{1-i}(C, p_i, c_r)$ and
let $p_{1-i}$ take a step to receive $mn_{1-i}(C, p_i, c_r)$.
Since $T_r$ is a fast transaction, once 
$p_{1-i}$ receives $mn_{1-i}(C, p_i, c_r)$, 
it sends a response $mn_{1-i}'(C, p_i, c_r)$ to $c_r$. % within a finite number of steps.
This sequence of steps starting from $C_{new}(C, p_i, c_r)$ to the step in which $p_{1-i}$ sends the response is denoted by $\sigma_{new}(C, p_i, c_r)$. 
Next, the adversary schedules the delivery of $mn_{i}(C, p_i, c_r)$ and  lets
$p_i$ take steps until it receives $mn_i(C, p_i, c_r)$
and sends a response to $c_r$. % within a finite number of steps.
Finally, $c_r$ takes steps to receive the responses from $p_0$ and $p_1$
and return a response for $T_r$.
\end{construction}
}

By the way $\gamma_{new}(C, p_i, c_r)$ is constructed, by 
Definition \ref{def:vis-val} and by Lemma~\ref{lem:granularity}, 
we get the following.
\vspace*{-.1cm}
\begin{observation}
\label{obs:gnew}
The following claims hold for $\gamma_{new}(C, p_i, c_r)$:
\begin{enumerate}
\item \label{gnew_legal} Execution $\gamma_{new}(C, p_i, c_r)$ is legal from $C$.
\y{\item \label{gnew_ind} Configurations $C$ and $RC(C, \gamma_{new}(C, p_i, c_r))$ are indistinguishable 
to $c_w$ and $p_i$}.
\item \label{gnew_return} The return value for $T_r$ in $\gamma_{new}(C, p_i, c_r)$ is  $(x_0, x_1)$. 
\end{enumerate}
\end{observation}
\vspace*{-.2cm}
\subsection{The Infinite Execution}
%
%\vspace*{-.1cm}
\y{Recall that we prove Theorem~\ref{thm:imp} by 
constructing an infinite execution $\alpha$ 
which contains just one write-only transaction $T_w$
and by proving that in $\alpha$, 
the values written by $T_w$ never becomes visible. 
We construct $\alpha$ in Lemma \ref{lma:induction} using induction . 
Specifically, the lemma shows that there is an infinite sequence of executions 
$\alpha_1,\alpha_2, \ldots$ such that, all of them are distinct prefixes of $\alpha$ 
(we let $\alpha_0$ be the empty execution).
Lemma \ref{lma:induction} is comprised of two claims which hold for every integer $k \geq 0$. 
The first, shows that $\alpha_k$ contains the transmission of at least one message 
which is sent after the execution of $\alpha_{k-1}$ from $C_0$ (and thus, 
$\alpha_{k-1}$ is a prefix of $\alpha_k$
and $\alpha_{k-1} \neq \alpha_k$).
The second shows that after $\alpha_k$ has been performed from $C_0$,
the two values written by $T_w$ have not yet become visible. 
%The proof of the second claim requires that the first claim holds.
Although the proofs of the two claims exhibit many similarities, 
for clarity of presentation, we have decided not to merge them into one proof.
}

\vspace*{-.1cm}
\begin{lemma}
\label{lma:induction}
For any integer $k\geq 1$, there exists an execution $\alpha_k$, legal from $C_0$, 
in which only one transaction $T_w=(w(X_0)x_0, w(X_1)x_1)$ is executed by client $c_w$.
Let $C_k$ be the configuration that results when $\alpha_k$ is applied from $C_0$.
Then, the following hold:
\begin{enumerate}
\item\label{claim:msg} $\alpha_k = \alpha_{k-1}\cdot\alpha_k'$, where in $\alpha_k'$ \dd{at least} one of the following occurs: 
\begin{itemize}
\item a message is sent by $p_{k\% 2}$ to $p_{(k-1)\% 2}$, or
\item a message is sent by $p_{k\% 2}$ to $c_w$ and it holds that after $c_w$ receives this message, $c_w$ sends a message to $p_{(k-1)\% 2}$.
\end{itemize}
\item\label{claim:vis} 
In $C_k$, $x_0$ and $x_1$ are not visible \y{and $T_w$ is still active}.
\end{enumerate}
\end{lemma}
\vspace*{-.3cm}
\begin{proof}[Proof \jjw{ of Lemma \ref{lma:induction}}]
By induction on $k$.  
%
%To avoid repetition, 
We prove the base case \y{together with} the induction step, 
yet we clearly state the difference when the proof diverges.
%In what follows, we prove the base case ($k=1$) and the induction step (for any integer $k\geq 2$).
%For the base case, we make no assumptions. 
To prove the induction step, fix an integer $k > 1$ 
and assume that the claim holds for any $j$, 
$1 \leq j < k$. 

\ifthenelse{\boolean{picturepage}}{}{
\begin{figure*}
    \centering
    \hfill
  \begin{subfigure}[b]{0.3\textwidth}
    \includegraphics[width=\textwidth]{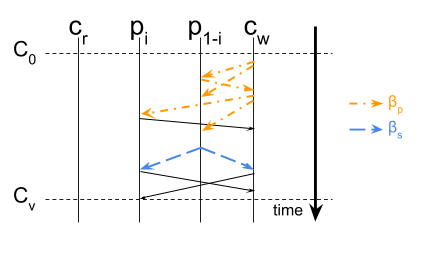}
    \caption{Execution $\beta$ and its subsequences}
    \label{fig:beta}
    \end{subfigure}\hfill
  \begin{subfigure}[b]{0.5\textwidth}
    \includegraphics[width=\textwidth]{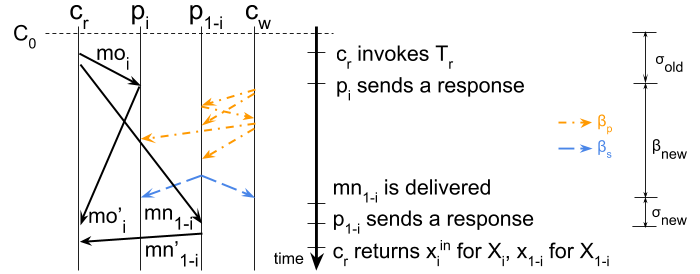}
    \caption{Execution $\gamma$}
    \label{fig:gamma}
    \end{subfigure}
    \hfill
\vspace*{-.2cm}
    \caption{Executions $\beta$ and $\gamma$. 
Execution $\beta$ consists of the steps taken by $c_w$, $p_i$ and $p_{1-i}$ from configuration $C_0$, where $x_0$ and $x_1$ are not visible yet, to $C_v$, where these two values become visible. }
    \label{fig:lma1-cons}
\vspace*{-.3cm}
\end{figure*}
}

{\bf Proof of claim~\ref{claim:msg}.}
\y{We start with claim \ref{claim:msg}. We start with a high-level description of the claim's proof.
The proof is by contradiction. }
We come up with two executions that have the following properties.
In the first execution, a read-only transaction $T_r = (r(X_0)*, r(X_1)*)$ 
(initiated by a client $c^k_r \not\in\{ c_0^{in}, c_1^{in} , c_w \}$)
is executed starting from $C_{k-1}$. \y{Since $x_0$ and $x_1$ are not visible
neither in $C_0$ (since $T_w$ has not yet started its execution in $C_0$), 
nor in $C_{k-1}$ if $k > 1$ (by induction hypothesis), 
the response for $T_r$ in this execution is $(x_0^{in}, x_1^{in})$}.
This execution is constructed based on Construction~\ref{gold_construction}.
In the second execution, $c^k_r$ invokes $T_r$ after $T_w$ has executed solo long enough
so that both values $x_0$ and $x_1$ are visible \y{(minimal progress for write-only
transactions implies that visibility of $x_0$ and $x_1$ will eventually happen)}. So, $c^k_r$ 
returns $(x_0, x_1)$ for $T_r$ in this execution. The second execution 
is constructed based on Construction~\ref{gnew_construction}.
We then combine parts of these two executions to get a third execution, $\gamma$.
Execution $\gamma$ is constructed so that we can prove that in it,
$c^k_r$ will return the same value for $X_{k \%2}$ as $c^k_r$ does in the first execution,
and the same value for $X_{(k-1) \%2}$ as $c^k_r$ does in the second execution. 
\y{Therefore, in $\gamma$, $c^k_r$ returns $(x_{k\% 2}^{in}, x_{(k-1)\%2})$ for $T_r$. 
This contradicts Lemma~\ref{lem:granularity}.}

We continue with the details of the proof of claim~\ref{claim:msg}. 
\y{To derive a contradiction, assume that the claim does not hold, 
i.e., we let $c_w, p_0, p_1$  take steps starting from $C_{k-1}$ and 
assume the following:}
\begin{itemize}
\item $p_{k\% 2}$ sends no message to $p_{(k-1)\% 2}$;
\item $p_{k\% 2}$ sends no message to $c_w$ for which it holds that after $c_w$ receives it, 
$p_{k\% 2}$ sends a message to $p_{(k-1)\% 2}$.
\end{itemize}

\y{
Since $x_0$ and $x_1$ are not visible neither in $C_0$ 
nor in $C_{k-1}$,
Lemma~\ref{lem:not-visible} implies that 
there exists at least one client
$c_r^k \not\in \{c_w, c_0^{in}, c_1^{in} \}$
such that if, starting from $C_{k-1}$, 
$c_r^k$ executes a read-only transaction 
$T_r = (r(X_0)*, r(X_1)*)$, then $c_r^k$ returns 
$(x_0^{in}, x_1^{in})$ for $T_r$.}
We derive the contradiction by constructing the execution $\gamma$, 
in which, in addition to $T_w$, $c_r^{k}$ %different from $c_w$ 
%$\not\in \{c_w, c_0^{in}, c_1^{in}\}$, 
executes such a read-only transaction $T_r$,
and by showing that $\gamma$ contradicts Lemma \ref{lem:granularity}.
%and $c_r^{k}$ returns from $T_r$. 

\y{
To construct $\gamma$, we need to define an execution $\beta$ and a subsequence
$\beta_{new}$ of it. 
Specifically, $\beta_{new}$ is utilized as part of execution $\gamma$. Roughly speaking,
$\gamma$ starts with $\sigma_{old}(C_{k-1}, p_{k \%2}, c_r^k)$ (see Construction~\ref{gold_construction})
up to the point that 
$p_{k\%2}$ reports $x_{k\%2}^{in}$ for $X_{k\%2}$ to $c_r^k$ (see Observation~\ref{obs:gold}). 
Recall that only processes $c_r^k$ and $p_{k\%2}$ take steps in $\sigma_{old}(C_{k-1}, p_{k \%2}, c_r^k)$
(by Observation~\ref{obs:gold}). 
Then, the events of $\beta_{new}$ are executed to take the system in a configuration
where $x_0$ and $x_1$ are visible. 
The assumption we made above to derive a contradiction,
allows us to design $\beta_{new}$ so that, $c_r^k$ and $p_{k \% 2}$ do not take any steps
in it, and the values written by $T_w$ are visible after $\beta_{new}$ is applied
starting from $C_{k-1}$ (as well as from $RC(C_{k-1}, \sigma_{old}(C_{k-1}, p_{k \%2}, c_r^k))$).
Afterwards, an execution $\gamma_{new}$, which is derived
based on Construction~\ref{gnew_construction}, is applied  
from the resulting configuration. In $\gamma_{new}$ only process
$p_{(k-1)\%2}$ take steps (see Observation~\ref{obs:gnew}).
Execution $\beta_{new}$ has been designed so that 
$p_{(k-1)\%2}$ is "unaware" of $p_{k\% 2}$'s decision
on what to report  to $c_r^k$ as the current value of the object that $p_{k\% 2}$ stores.
So, $p_{(k-1)\%2}$ reports $x_{(k-1)\%2}$ for $X_{(k-1)\%2}$ as it does in Construction~\ref{gnew_construction}
(Observation~\ref{obs:gnew}). 
The construction of $\gamma$ concludes with $c_r^k$ taking steps until $T_r$ responds. 
We argue (below) 
that $c_r^k$ receives in $\gamma$
the same message regarding the value of $X_{k\%2}$ as in $\gamma_{old}(C_{k-1}, p_{k \%2}, c_r^k)$, 
so it returns $x_{k\%2}^{in}$ for $X_{k\%2}$ (by Observation~\ref{obs:gold}).
We also argue that  $c_r^k$ receives in $\gamma$
the same message regarding the value of $X_{(k-1)\%2}$ as in 
$\gamma_{new}$, 
so it returns $x_{(k-1)\%2}$ for $X_{(k-1)\%2}$ (by Observation~\ref{obs:gnew}). 
This contradicts Lemma \ref{lem:granularity}.
}

% (or from any configuration that is indistinguishable  from $C_{k-1}$ to $c_r^k$ and $p_{(k-1) \% 2}$.}

We first define $\beta$.
\y{For the base case (where $k = 1$),} $c_w$ invokes $T_w$ starting from $C_0$ and executes solo
(i.e., only $c_w, p_0, p_1$ take steps) until $x_0$ and $x_1$ are visible; 
since $T_w$ has not yet been invoked in $C_0$, $x_0$ and $x_1$ are not visible in $C_0$.
\y{For the induction step, the induction hypothesis (claim~\ref{claim:vis}) implies that
in $C_{k-1}$, $x_0$ and $x_1$ are not visible. %and $T_w$ is active. 
Again, we let $c_w$  execute solo, starting from $C_{k-1}$,
until $x_0$ and $x_1$ are visible (minimal progress implies that this will eventually happen).}
In either case, let $C_v$ be the first configuration after $C_{k-1}$ in which $x_0$ and $x_1$ are visible. 
Let $\beta$ be the sequence of steps taken from $C_{k-1}$ to $C_v$
\y{(all of them are by $c_w$ and the servers). }

\y{In this and the next two paragraphs, we define $\beta_{new}$ and show that it is legal from $C_{k-1}$.
Let $\beta'_p$ be the shortest prefix of $\beta$ which contains all messages sent by $c_w$ to $p_{(k-1)\% 2}$,
and let $\beta'_s$ be the remaining suffix of $\beta$.
Let $\beta_p$ be the subsequence of $\beta'_p$ in which all steps taken by $p_{k\% 2}$ have been removed.
Let $\beta_s$ be the subsequence of $\beta'_s$ containing only steps by $p_{(k-1)\% 2}$.
Let $\beta_{new}$ be $\beta_p\cdot \beta_s$.}
\y{Note that $\beta_{new}$ does not contain any step by $p_{k \% 2}$.}
\jjw{Executions $\beta$, $\beta_p$ and $\beta_s$ are illustrated in Figure \ref{fig:beta}. In the figure, symbols $i, 1-i\in\{0,1\}$ refer to $k\%2$ and $(k-1)\%2$ respectively.}

To show that $\beta_{new}$ is legal from $C_{k-1}$, we first argue that $\beta_p$ is legal from $C_{k-1}$.
Because $\beta'_p$ is a prefix of $\beta$, $\beta'_p$ is legal from $C_{k-1}$.
By assumption, \y{$p_{k\% 2}$ sends no message to $p_{(k-1)\% 2}$,
so if $p_{(k-1)\% 2}$ receives any message 
from $p_{k\% 2}$ in $\beta'_p$, then the message must have been sent before $C_{k-1}$ 
(and must have been received after it), 
i.e., the message is not sent in $\beta'_p$.} 
\y{(For the base case, since no message is in transit in $C_0$, $p_{(k-1)\% 2}$ receives no message 
from $p_{k\% 2}$ in $\beta'_p$.)}
Moreover,  
by assumption, after $\alpha_{k-1}$,
$p_{k\% 2}$ sends no message to $c_w$ for which it holds that after $c_w$ receives it, $c_w$ sends a message to $p_{(k-1)\% 2}$.
Since, by definition, $\beta'_p$ ends with a message sent by $c_w$ to $p_{(k-1)\% 2}$ (if $\beta'_p$ is not empty),
it follows that:
\jjw{for the base case, $c_w$ receives no message from $p_{k\% 2}$ in $\beta'_p$; for the induction step, }
if $c_w$ receives any message from $p_{k\% 2}$ in $\beta'_p$, 
then the message has been sent \y{before $C_{k-1}$}.
\y{Thus, $\beta_p$, which results from
the removal of all steps taken by $p_{k\% 2}$ from $\beta'_p$, is legal from $C_{k-1}$.}
%So, $\beta_p$ is legal from $C_{k-1}$. 
Moreover,
$RC(C_{k-1}, \beta'_p)$ and $RC(C_{k-1}, \beta_p)$ \y{(i.e., the configurations that result
when $\beta'_p$ and $\beta_p$, respectively, are applied from $C_{k-1}$)}
are indistinguishable to $p_{(k-1)\% 2}$ and $c_w$.

To complete the argument that $\beta_{new}$ is legal from $C_{k-1}$,
it remains to prove that $\beta_s$ is legal from $RC(C_{k-1}, \beta_p)$. 
\y{By definition,} only $p_{(k-1)\% 2}$ takes steps in $\beta_s$.
\y{Note that, because $RC(C_{k-1}, \beta'_p)$ and $RC(C_{k-1}, \beta_p)$ are indistinguishable to $p_{(k-1)\% 2}$,  
by proving that $\beta_s$ is legal from $RC(C_{k-1}, \beta_p')$,
it follows that} $\beta_s$ is legal from $RC(C_{k-1}, \beta_p)$.
We next argue that $\beta_s$ is indeed legal from configuration $RC(C_{k-1}, \beta_p')$. 
By assumption, 
if $p_{(k-1)\% 2}$ receives any message from $p_{k\% 2}$ in $\beta_s$,
then the message has been sent before $C_{k-1}$ \y{(i.e., the message has not been sent in $\beta$)}.
\y{For the base case, since no message is in transit in $C_0$, 
$p_{(k-1)\%2}$ receives no message from $p_{k\%2}$ in $\beta_s$.}
Recall that, by definition of $\beta'_p$, all messages from $c_w$ to $p_{(k-1)\% 2}$ are sent by the end of $\beta_p'$.
Therefore, $c_w$ does not send any message to $p_{(k-1)\% 2}$ in $\beta'_s$. 
So, any message that $p_{(k-1)\% 2}$ receives from $c_w$ in $\beta_s$ has been sent by the end of $\beta_p'$. 
Thus, $\beta_s$ is legal from $RC(C_{k-1}, \beta_p')$,
\y{and since $RC(C_{k-1}, \beta'_p)$ and $RC(C_{k-1}, \beta_p)$ are indistinguishable to $p_{(k-1)\% 2}$},
$\beta_s$ is also legal from $RC(C_{k-1}, \beta_p)$. 
\y{Therefore, $\beta_{new} = \beta_p \cdot \beta_s$ is legal from $C_{k-1}$.}
%Moreover, 
%by arguing that $\beta_s$ is legal from $RC(C_{k-1}, \beta_p')$,

\y{From the arguments above, it also follows that} 
\y{$RC(C_{k-1}, \beta_p\cdot\beta_s)$ and $C_v = RC(C_{k-1}, \beta_p'\cdot\beta_s')$ 
are indistinguishable to $p_{(k-1)\% 2}$. 
%Because $RC(C_{k-1}, \beta'_p)$ and $RC(C_{k-1}, \beta_p)$ are indistinguishable to $p_{(k-1)\% 2}$,
%it follows that 
Therefore,
$RC(C_{k-1}, \beta_{new})$ and $C_v$ are indistinguishable to $p_{(k-1)\% 2}$. 

\y{We continue to construct $\gamma$. To do so,} 
we use $\sigma_{old}(C_{k-1}, p_{k\% 2}, c_r^{k})$ (Construction~\ref{gold_construction}),
$\beta_{new}$, 
and $\sigma_{new}(C_v, p_{k\% 2}, c_r^{k})$   (Construction~\ref{gnew_construction}).
We also refer to configuration $C_{new}(C_v, p_{k\% 2}, c_r^{k})$  (Construction~\ref{gnew_construction}).
Figure \ref{fig:gamma} illustrates the construction of $\gamma$, 
where symbols $i, 1-i\in\{0,1\}$ refer to $k\%2$ and $(k-1)\%2$ respectively.}
(For simplicity, we omit $(C_{k-1}, p_{k\% 2}, c_r^{k})$
and $(C_v, p_{k\% 2}, c_r^{k})$ from the notations below.)

Recall that in $\gamma$, a client starts a read-only transaction \y{from $C_{k-1}$. 
One server responds to the client first (i.e., $\sigma_{old}$ is applied from $C_{k-1}$). 
Then the write-only transaction makes progress and the values written 
turn to be visible (specifically, $\beta_{new}$ is applied after $\sigma_{old}$). 
Then the other server receives the request of the read-only transaction and responds to the client (specifically, 
$\sigma_{new}$ is applied).}
Recall that (as we argue below) to one server, 
$\gamma$ is indistinguishable from $\gamma_{old}$ (i.e., the execution illustrated 
in Figure \ref{fig:gamma} is indistinguishable from that in Figure \ref{fig:gold}) 
and thus the server returns an old value,
while to the other server, $\gamma$ is indistinguishable from $\gamma_{new}$ 
(i.e., the execution illustrated in Figure \ref{fig:gamma} is indistinguishable 
from that in Figure \ref{fig:gnew}) and thus the other server returns a new value. 
This then leads to the contradiction.
%The indistinguishability argument is fundamental for the construction of $\gamma$ 
%(so that all steps in $\gamma$ are legal) as well as for the analysis of the return value of the read-only transaction.

\y{We are now ready to formally define $\gamma$.
Starting from $C_{k-1}$, the adversary applies 
$\sigma_{old}(C_{k-1}, p_{k\% 2}, c_r^{k}) \cdot
 \beta_{new} \cdot 
 \sigma_{new}(C_v, p_{k\% 2}, c_r^{k}))$
(we later prove that the application of these steps from $C_{k-1}$ is legal).}
% sentences are re-ordered
%\jjw{Then $\beta_{new}$ is applied from $RC(C_{k-1}, \sigma_{old})$.
%Next we apply $\sigma_{new}$ from 
%$RC(C_{k-1},  \sigma_{old}\cdot\beta_{new})$.
\y{By Construction~\ref{gold_construction}, 
in the last step of $\sigma_{old}$, $p_{k\% 2}$ sends a message $mo_{k\% 2}'$ to $c_r^{k}$. 
Similarly, by Construction~\ref{gnew_construction}, 
in the last step of $\sigma_{new}$, $p_{(k-1)\% 2}$ sends a message $mn_{(k-1)\% 2}'$ to $c_r^{k}$.}
The adversary next schedules the delivery of $mo_{k\% 2}'$ and $mn_{(k-1)\% 2}'$, and lets
$c_r^{k}$ take steps until $T_r$ completes (this will happen because $T_r$ is a fast transaction). 
This concludes the construction of $\gamma$.

%Recall that in the last step of $\sigma_{old}$, $p_{k\% 2}$ sends a message $mo_{k\% 2}'$ to $c_r^{k}$.
\y{We now argue that $\gamma$ is legal.}
\y{By construction, only processes $c_r^k$ and $p_{k\% 2}$ take steps in $\sigma_{old}$.
By Observation~\ref{obs:gold} (claim~\ref{gold_ind}),
$C_{k-1}$ and $RC(C_{k-1}, \sigma_{old})$ are indistinguishable to $c_w$ and $p_{(k-1)\% 2}$.
Since only $c_w$ and $p_{(k-1)\% 2}$ 
take steps in $\beta_{new}$, it follows that $\beta_{new}$ is legal from $RC(C_{k-1}, \sigma_{old})$. } 
Since the processes that take steps in $\sigma_{old}$ and $\beta_{new}$ are disjoint,
then $RC(C_{k-1}, \sigma_{old} \cdot \beta_{new})$ and $RC(C_{k-1},  \beta_{new}\cdot \sigma_{old} )$ are indistinguishable
\y{to all processes.
Recall that  $RC(C_{k-1}, \beta_{new})$ and $C_v$ are indistinguishable to $p_{(k-1)\% 2}$.}
\y{Since $\sigma_{old}$ is composed of a sequence of steps in which only $c_r^k$ 
and $p_{k\%2}$ take steps, $RC(C_v, \sigma_{old})$ 
and $C_{new}(C_v, p_{k\% 2}, c_r^{k})$
are indistinguishable to $p_{(k-1)\% 2}$. It follows that }
$RC(C_{k-1},  \beta_{new}\cdot \sigma_{old} )$ and $C_{new}(C_v, p_{k\% 2}, c_r^{k})$
are indistinguishable to $p_{(k-1)\% 2}$.
Because only $p_{(k-1)\% 2}$ takes steps in $\sigma_{new}$
\y{ and $\sigma_{new}$  is legal from $C_{new}$ (by definition),}
it follows that $\sigma_{new}$ is legal 
%from both $RC(C_{k-1},  \beta_{new}\cdot \sigma_{old} )$ and 
from $RC(C_{k-1},  \sigma_{old}\cdot\beta_{new})$.
%\jjw{Thus we can apply $\sigma_{new}$ from  $RC(C_{k-1},  \sigma_{old}\cdot\beta_{new})$.
%As a result, the construction of $\gamma$ is legal.}
Therefore, $\gamma$ is legal.

We now focus on the values returned by $c_r^k$ for $T_r$ in $\gamma$.
In $\gamma$, $c_r^{k}$ executes only transaction $T_r$.
Thus, $c_r^{k}$ decides the response for $T_r$ based solely on the values included in $mo_{k\% 2}'$ 
\y{(sent by $p_{k \%2}$ in $\sigma_{old}$)} and $mn_{(k-1)\% 2}'$ \y{(sent by $p_{(k-1)\%2}$ in $\sigma_{new}$)}.
For the base case, $mo_{k\% 2}'$ is sent before $c_w$ takes any step, 
so $mo_{k\% 2}'$ contains neither $x_0$ nor $x_1$. 
For the induction step, since $mo_{k\% 2}'$ is sent in $\gamma_{old}$, 
by Observation \ref{obs:gold} and the one-value messages property, 
$mo_{k\% 2}'$ contains neither $x_0$ nor $x_1$. 
Recall that $mn_{(k-1)\% 2}'$ is sent in $\gamma_{new}$.
By Observation \ref{obs:gnew} and the one-value messages property, 
$mn_{(k-1)\% 2}'$ contains the value $x_{(k-1)\% 2}$. 
\y{Recall that (by assumption) $c_r^k$ does not receive any other messages
from $p_0$ and $p_1$.}
\y{Thus, $c_r^k$ receives for $X_{k\%2}$ just one value, 
namely, the same value it receives for it in $\gamma_{old}$.
Similarly, it receives for $X_{(k-1)\%2}$ just one value,
namely, the same value it receives for it in $\gamma_{new}$.}
It follows that in $\gamma$, the values that $c_r^k$
returns are 
$x_{k\% 2}^{in}$ for $X_{k\% 2}$ 
and $x_{(k-1)\% 2}$ for $X_{(k-1)\% 2}$. 
This contradicts Lemma \ref{lem:granularity}. 
\y{(Note that  $\gamma$ is an execution utilized just for proving claim~\ref{claim:msg};
for every $k > 1$, we build it from scratch to prove the induction step for $k$.)}

\noindent
{\bf Definition of $\alpha_k$ and $C_k$.}
\y{We now define $\alpha_k$ and $C_k$.}
By claim~\ref{claim:msg},  it follows that 
in any legal execution starting from $C_{k-1}$
in which $c_w$ executes solo, 
at least one of the following two statements hold: 
% in which  
%just one transaction $T_w = (w(X_0)x_0, $ $w(X_1)x_1)$ is active:
(1) $p_{k\%2}$ sends a message to $p_{(k-1)\%2}$;
(2) $p_{k\%2}$ sends a message to $c_w$ so that after $c_w$ receives this message, $c_w$ sends a message to $p_{(k-1)\%2}$.
Let $ms_{k}$ be the first message that satisfies any of the two statements above.
We construct execution $\alpha_k'$ as follows. 
In $\alpha_k'$, $T_w$ executes solo starting from $C_{k-1}$ until $ms_k$ 
is sent\footnote{Clearly, in $\alpha_k'$, $p_{k\%2}$ must take 
at least one step. For the case where $k\geq 2$, we also require that 
in $\alpha_k'$, in the first step which $p_{k\% 2}$ takes, 
message $ms_{k-1}$ is delivered at $p_{k\% 2}$ (in order 
to comply with our model of finite message delay).}.
\y{Let $\alpha_k = \alpha_{k-1} \cdot \alpha_k'$, and let $C_k$ be the configuration 
that results when $\alpha_k$ is applied from $C_0$.}

\noindent
{\bf Proof of claim~\ref{claim:vis}.}
\jjw{To prove claim~\ref{claim:vis}, 
we use similar arguments as those in the proof of claim~\ref{claim:msg}. 
We assume that claim~\ref{claim:vis} does not hold, 
i.e., we assume that in $C_{k}$, $x_i$ is visible for some $i\in\{0,1\}$. 
To derive a contradiction, we construct an execution $\delta$ \y{(in a way similar to that 
we construct $\gamma$)} 
and show that $\delta$ contradicts Lemma~\ref{lem:granularity}. 
We first define executions $\rho$ and $\rho_{new}$ 
in a way similar to $\beta$ and $\beta_{new}$  defined in the proof of 
claim~\ref{claim:msg}.
We finally define $\delta$ based on $\sigma_{old}(C_{k-1}, p_{k\%2},c_r^{k})$, $\rho_{new}$, 
and $\sigma_{new}(C_k,p_{k\%2},c_r^{k})$, 
and we argue that in $\delta$, $c_r^{k}$ returns a response for $T_r$
that contradicts Lemma~\ref{lem:granularity}.    }

We now present the details of the proof \jjw{of claim~\ref{claim:vis}}. 
Assume that in $C_k$, $x_i$ is visible for some $i\in\{0,1\}$.
We construct $\delta$ by utilizing executions $\sigma_{old}(C_{k-1}, p_{k\%2}, c_r^{k})$ 
and $\gamma_{old}(C_{k-1}, p_{k\%2}, c_r^{k})$
(Construction~\ref{gold_construction}),
as well as $\sigma_{new}(C_k, p_{k\%2}, c_r^{k})$
and $\gamma_{new}(C_k, p_{k\%2}, c_r^{k})$  (Construction~\ref{gnew_construction}). 
By Observation~\ref{obs:gold}, $c_r^{k}$ returns
$x_{(k-1)\%2}^{in}$ for $X_{(k-1)\%2}$
and $x_{k\%2}^{in}$ for $X_{k\%2}$
for $T_r$ in $\gamma_{old}(C_{k-1}, p_{k\%2}, c_r^{k})$.
Because by assumption, $x_i$ is visible at $C_k$, Observation~\ref{obs:gnew} implies that 
$c_r^{k}$ returns $x_{(k-1)\%2}$  for $X_{(k-1)\%2}$ and $x_{k\%2}$ for $X_{k\%2}$ for $T_r$ in $\gamma_{new}(C_k, p_{k\%2}, c_r^{k})$.

To construct $\delta$, we first define an execution $\rho$ and some subsequences of it,
and we study their properties. \y{(The construction of $\rho$ and its subsequences 
is similar to that of $\beta$ and its subsequences. 
Yet, the reasoning of why the construction is legal is different, and thus worth-presenting.)}
Let $\rho$ be the sequence of steps which are taken from $C_{k-1}$ to $C_k$,
i.e., $\rho = \alpha_k'$.
Let $\rho'_p$ be the shortest prefix of $\rho$ which contains all messages sent by $c_w$ to $p_{(k-1)\%2}$,
and let $\rho'_s$ be the remaining suffix of $\rho$.
Let $\rho_p$ be the subsequence of $\rho'_p$ in which all steps taken by $p_{k\%2}$ have been removed.
Let $\rho_s$ be the subsequence of $\rho'_s$ containing only steps by $p_{(k-1)\%2}$.
Let $\rho_{new}$ be $\rho_p\cdot \rho_s$.
We utilize $\rho_{new}$ as part of our construction of $\delta$ below (\y{as we did with $\beta_{new}$ and $\gamma$). 
In the next two paragraphs, we} argue that $\rho_{new}$ is legal from $C_{k-1}$.

We first argue that $\rho_{p}$ is legal from $C_{k-1}$.
Because $\rho'_p$ is a prefix of $\rho$, $\rho'_p$ is legal from $C_{k-1}$.
\y{For the base case ($k = 1$), since no message is in transit in $C_0$, 
the definition of $\alpha_k$ implies that
$p_{(k-1)\% 2}$  receives no message from $p_{k\% 2}$ in $\rho'_p$.
For the induction step,}
the definition of $\alpha_k$ implies that if $p_{(k-1)\% 2}$ receives any message 
from $p_{k\% 2}$ in \y{$\rho'_p$}, then the message has been sent before $C_{k-1}$, 
i.e., the message is not sent in $\rho'_p$.
Moreover,  
by definition of $\alpha_k'$, % ends with the first message sent by $p_{k\%2}$ to $p_{(k-1)\%2}$ 
%\y{(explicitly or implicitly through $c_w$),}
$p_{k\%2}$  sends no message to $c_w$ for which it holds that after the receipt of this message, 
$c_w$ sends a message to $p_{(k-1)\%2}$.
Since, by definition, $\rho'_p$ ends with a message sent by 
$c_w$ to $p_{(k-1)\%2}$ (if $\rho'_p$ is not empty), it follows that:
\y{for the base case, $c_w$ receives no message from $p_{k\% 2}$ in $\rho'_p$;
for the induction step, }
if $c_w$ receives any message from $p_{k\% 2}$ in $\rho'_p$, 
then the message has been sent before $C_{k-1}$. %, i.e., before the prefix $\rho'_p$.
Thus, \y{$\rho_p$, which results from the removal of all steps taken by $p_{k \% 2}$
from $\rho'_p$,} is legal from $C_{k-1}$. Moreover,
$RC(C_{k-1}, \rho'_p)$ and $RC(C_{k-1}, \rho_p)$ are indistinguishable to $p_{(k-1)\%2}$ and $c_w$.

\y{To  complete the argument that $\rho_{new}$ is legal from $C_{k-1}$,
it remains to argue that $\rho_s$ is legal from $RC(C_{k-1}, \rho_p)$. }
\y{By definition, only $p_{(k-1)\%2}$ takes steps in $\rho_s$.
note that, because $RC(C_{k-1}, \rho'_p)$ and $RC(C_{k-1}, \rho_p)$ 
are indistinguishable to $p_{(k-1)\%2}$,  
by proving that $\rho_s$ is legal from $RC(C_{k-1}, \rho_p')$,
it follows that $\rho_s$ is legal from $RC(C_{k-1}, \rho_p)$.
We next argue that $\rho_s$ is indeed legal from configuration 
$RC(C_{k-1}, \rho_p')$. 
By the definition of $\alpha_k$, 
if $p_{(k-1)\% 2}$ receives any message from $p_{k\% 2}$ in $\rho_s$,
then the message has been sent before $C_{k-1}$ (i.e., the message
has not been sent in $\rho$).
For the base case, since no message is in transit in $C_0$,
the definition of $\alpha_k$ implies that 
$p_{(k-1)\% 2}$ receives no message from $p_{k\% 2}$ in $\rho_s$.}
Recall that by definition of $\rho_s'$, 
all messages from $c_w$ to $p_{(k-1)\%2}$ are sent by the end of $\rho_p'$).
%so any message $p_{(k-1)\%2}$ receives from $c_w$ in $\rho_s$ has been sent in or before $\rho_p'$. 
Thus, $\rho_s$ is legal from $RC(C_{k-1}, \rho_p')$,
and since $RC(C_{k-1}, \rho_p')$ and $RC(C_{k-1}, \rho_p)$ are indistinguishable to $p_{(k-1)\%2}$,
$\rho_s$ is also legal from $RC(C_{k-1}, \rho_p)$. 
Therefore, $\rho_{new} = \rho_p \cdot \rho_s$ is legal from $C_{k-1}$.

From the arguments above, it also follows that 
$RC(C_{k-1}, \rho_p\cdot\rho_s)$ and $C_k = RC(C_{k-1}, \rho_p'\cdot\rho_s')$ are indistinguishable to $p_{(k-1)\%2}$. 
Because $RC(C_{k-1}, \rho'_p)$ and $RC(C_{k-1}, \rho_p)$ are indistinguishable to $p_{(k-1)\%2}$,
it follows that 
$RC(C_{k-1}, \rho_{new})$ and $C_k$ are indistinguishable to $p_{(k-1)\%2}$.

\y{
We are now ready to formally define $\delta$.
To do so, we use executions 
$\sigma_{old}(C_{k-1}, p_{k\%2}, c_r^{k})$
(Construction~\ref{gold_construction}),
$\sigma_{new}(C_k, p_{k\%2}, c_r^{k})$  (Construction~\ref{gnew_construction}) 
and $\rho_{new}$.}
We also refer to $C_{new}(C_k, p_{k\%2}, c_r^{k})$  (Construction~\ref{gnew_construction}). 
\y{
Starting from $C_{k-1}$, the adversary applies the step sequence
$\sigma_{old}(C_{k-1}, p_{k\% 2}, c_r^{k}) \cdot
 \rho_{new} \cdot 
 \sigma_{new}(C_v, p_{k\% 2}, c_r^{k}))$
(we later prove that the application of these steps from $C_{k-1}$ is legal).}
For simplicity, we omit $(C_{k-1}, p_{k\% 2}, c_r^{k})$
and $(C_v, p_{k\% 2}, c_r^{k})$ from the notations below
(thus abusing notations 
$\sigma_{new}$ and $C_{new}$ which were also used in the proof of claim~\ref{claim:msg}.)
% sentences are re-ordered
%\jjw{Then $\beta_{new}$ is applied from $RC(C_{k-1}, \sigma_{old})$.
%Next we apply $\sigma_{new}$ from 
%$RC(C_{k-1},  \sigma_{old}\cdot\beta_{new})$.
\y{By Construction~\ref{gold_construction}, 
in the last step of $\sigma_{old}$, $p_{k\% 2}$ sends a message $mo_{k\% 2}'$ to $c_r^{k}$. 
Similarly, by Construction~\ref{gnew_construction}, 
in the last step of $\sigma_{new}$, $p_{(k-1)\% 2}$ sends a message $mn_{(k-1)\% 2}'$ to $c_r^{k}$.}
The adversary next schedules the delivery of $mo_{k\% 2}'$ and $mn_{(k-1)\% 2}'$, and lets
$c_r^{k}$ take steps until $T_r$ completes (this will happen because $T_r$ is a fast transaction). 
This concludes the construction of $\delta$.
% sentences are re-organized

%Recall that in the last step of $\sigma_{old}$, $p_{k\%2}$ sends message $mo_{k\%2}'$ to $c_r^{k}$.
\jjw{We now argue that $\delta$ is legal.}
\y{By construction, only processes $c_r^k$ and $p_{k\% 2}$ take steps in $\sigma_{old}$.
By Observation~\ref{obs:gold} (claim~\ref{gold_ind}),
$C_{k-1}$ and $RC(C_{k-1}, \sigma_{old})$ are indistinguishable to $c_w$ and $p_{(k-1)\% 2}$.
Since only $c_w$ and $p_{(k-1)\% 2}$ 
take steps in $\rho_{new}$, it follows that $\rho_{new}$ is legal from $RC(C_{k-1}, \sigma_{old})$. }
Since the processes that take steps in $\sigma_{old}$ and $\rho_{new}$ are disjoint,
then $RC(C_{k-1}, \sigma_{old} \cdot \rho_{new})$ and $RC(C_{k-1},  \rho_{new}\cdot \sigma_{old} )$ are indistinguishable
\y{to all processes.
Recall that  $RC(C_{k-1}, \rho_{new})$ and $C_k$ are indistinguishable to $p_{(k-1)\% 2}$.}
\y{Since $\sigma_{old}$ is a composed of sequence of steps in which only $c_r^k$ 
and $p_{k\%2}$ take steps, $RC(C_k, \sigma_{old})$ 
and $C_{new}(C_k, p_{k\% 2}, c_r^{k})$
are indistinguishable to $p_{(k-1)\% 2}$. It follows that }
$RC(C_{k-1},  \rho_{new}\cdot \sigma_{old} )$ and $C_{new}(C_v, p_{k\% 2}, c_r^{k})$
are indistinguishable to $p_{(k-1)\% 2}$.
Because only $p_{(k-1)\% 2}$ takes steps in $\sigma_{new}$
\y{ and $\sigma_{new}$  is legal from $C_{new}$ (by definition),}
it follows that $\sigma_{new}$ is legal 
%from both $RC(C_{k-1},  \rho_{new}\cdot \sigma_{old} )$ and 
from $RC(C_{k-1},  \sigma_{old}\cdot\rho_{new})$.
%\jjw{Thus we can apply $\sigma_{new}$ from  $RC(C_{k-1},  \sigma_{old}\cdot\rho_{new})$.
%As a result, the construction of $\gamma$ is legal.}
Therefore, $\delta$ is legal.

We now focus on the values returned by $c_r^k$ for $T_r$ in $\delta$.
In $\delta$, $c_r^{k}$ executes only transaction $T_r$.
Thus, $c_r^{k}$ decides the response for $T_r$ based solely on the values included in $mo_{k\% 2}'$ 
\y{(sent by $p_{k \%2}$ in $\sigma_{old}$)} and $mn_{(k-1)\% 2}'$ \y{(sent by $p_{(k-1)\%2}$ in $\sigma_{new}$)}.
For the base case, $mo_{k\% 2}'$ is sent before $c_w$ takes any step, 
so $mo_{k\% 2}'$ contains neither $x_0$ nor $x_1$. 
For the induction step, since $mo_{k\% 2}'$ is sent in $\sigma_{old}$, 
by Observation \ref{obs:gold} and the one-value messages property, 
$mo_{k\% 2}'$ contains neither $x_0$ nor $x_1$. 
Recall that $mn_{(k-1)\% 2}'$ is sent in $\sigma_{new}$.
By Observation \ref{obs:gnew} and the one-value messages property, 
$mn_{(k-1)\% 2}'$ contains the value $x_{(k-1)\% 2}$. 
\y{Recall that (by assumption) $c_r^k$ does not receive any other messages
from $p_0$ and $p_1$.}
\y{Thus, $c_r^k$ receives for $X_{k\%2}$ just one value, 
namely, the same value it receives for it in $\gamma_{old}$ (Construction~\ref{gold_construction}).
Similarly, it receives for $X_{(k-1)\%2}$ just one value,
namely, the same value it receives for it in $\gamma_{new}$ (Construction~\ref{gnew_construction}).}
It follows that in $\delta$, the values that $c_r^k$
returns are 
$x_{k\% 2}^{in}$ for $X_{k\% 2}$ 
and $x_{(k-1)\% 2}$ for $X_{(k-1)\% 2}$. 
This contradicts Lemma \ref{lem:granularity}.
\end{proof}
\vspace{-.4cm}
\subsection{The Limits of the Impossibility Result}
\label{sec:limits}
\begin{comment}
Theorem~\ref{thm:imp} implies that multi-object write transactions (W)  are incompatible with 
read-only transactions that are
nonblocking (N), take one round of communication to complete (R), and in which servers return to clients
only one value for each read object (V). 
\end{comment}
Theorem~\ref{thm:imp} shows that multi-object write transactions (W)  are incompatible with nonblocking (N), one-roundtrip (O) and one-value (V) read-only transactions. 
In this section, we investigate the limits of our impossibility result. 
We show that it is sufficient to relax any of these properties to obtain a 
distributed storage system that satisfies the rest. 
To this end, we describe  possible designs  that achieve combinations
 of three out of the four properties.
%In our discussion, we focus on designs that achieve causal consistency, 
%even though, for some property combinations, there exist systems 
%that achieve these combinations while 
%ensuring other consistency levels (as we show in Section~\ref{rw_short}).

\vspace{-9pt}
~\\\noindent{\bf  N + R + V.} This combination supports fast 
read-only transactions and is implemented by COPS-SNOW~\cite{lu_snow_2016}. 
When a client $c$ writes a new value $x_1$ of object $X_1$, 
$c$ piggybacks the information about its causal dependencies.%, 
%i.e., the list of objects that $c$ has read ever since its last write operation.
Before making $x_1$ visible, the server $p_1$ storing $X_1$ contacts all servers 
that store objects listed in such dependency list. For each such object $X$, $p_1$ collects the identifiers of the 
read-only transactions that have read a value of $X$ that is not the last written. 
Then $p_1$ enforces that $x_1$ is invisible to these
read-only transactions. This prevents a read-only transaction from reading
$x_1$ and then $x_0$, if in the meanwhile $x'_0$ has been created such that $x_0 <^c x'_0 <^c x_1$.
Recall that COPS-SNOW does not ensure the W property, i.e., it does not 
support multi-object write transactions.

\vspace{-9pt}
~\\\noindent{\bf  N + V + W.} This design is implemented by Wren~\cite{Spirovska:2018}.
In this system, the servers periodically exchange information about the minimum timestamp
among those of complete transactions. This {\em cutoff} timestamp is such that 
there does not exist any non-complete (or future) transaction 
with a lower timestamp. The cutoff timestamp is used to identify a snapshot of the data
storage system from which 
a read-only transaction can read without blocking. 
 A new object written by a client is assigned a timestamp
 higher than the cutoff timestamp, so as to reflect the causal dependencies of the object. 
 Therefore, each client caches locally the values of the objects
  it  writes, as long as their timestamps are smaller than the cutoff. 
This mechanism allows a client  to read its own writes that are not included yet in the snapshot 
identified by the cutoff timestamp. Thus, each read-only transaction undergoes a first
 round of communication to get informed
about the cutoff timestamp (we remark that this timestamp can be provided by any server) 
and then executes a second round of communication to actually read the objects.  
%The read-only transaction reads an object directly from the client cache if there 
%exist a value of such object in the cache with a timestamp higher than 
%the cutoff timestamp used by the transaction.

\vspace{-9pt}
~\\\noindent{\bf  N + R + W.} Although we are not aware of any system that implements this design,
we briefly discuss a modification of the COPS system~\cite{lloyd_settle_2011} to achieve 
this combination of properties. 
COPS does not implement multi-object write transactions and implements 
read-only transactions
that are nonblocking but may require two rounds of communication, each communicating
just one value of the object to be read. We can augment COPS to achieve R and W as follows. 
Each write operation within a transaction must carry a) the values of the other objects written
in the same transaction and b) information about {\em all} objects on which the transaction 
causally depends (including their values).
 This additional meta-data is stored with each written object. %When a client $c$
  %reads an object $o$, the server responsible for $o$
   %returns to $c$ the last written value of $o$ together with the corresponding 
   %meta-data.
     Hence, $c$ executes a read-only transaction as follows. First, for each
    object $o$ to read, $c$ retrieves  the value of $o$ and the additional 
meta-data from the corresponding server.
     Then, once $c$ has received a reply from each involved partition, 
 $c$ identifies, for each object, the last written value, which is
      returned to the application. This protocol is not efficient, 
as it requires to store and communicate a prohibitively big amount of data. 
It is an open problem whether a more efficient N+R+W protocol exists. 

\vspace{-9pt}
~\\\noindent{\bf  R + V + W.} This design is implemented by RoCoCo-SNOW~\cite{lu_snow_2016} and Spanner~\cite{Corbett:2013}, 
which achieve strict serializability, and hence satisfy causal consistency. 
\begin{comment}
RoCoCo-SNOW, however, assumes transactions that are executed as {\em stored procedures}, 
whose logic is known {\em a priori}. For example, in the context of a distributed 
database, this enables RoCoCo-SNOW to know beforehand which database tables and columns
are touched by a transaction. RoCoCo-SNOW leverages this knowledge to re-order transactions
and disseminate dependency information that is needed to achieve the R+V+W properties 
(with a mechanism that is similar to the one employed by COPS-SNOW). 
\end{comment}
RoCoCo-SNOW implements a mechanism similar to COPS-SNOW, but assumes the {\em a priori} knowledge of the data accessed by transactions, which are executed as {\em stored procedures}. 
Spanner assumes tightly synchronized physical clocks and leverages  the known bound on clock drift
 to order transactions. 
It is an open problem whether a R + V + W implementation exists that does not rely on
 such assumptions.

\vspace*{-.2cm}
\section{Related Work}
\label{rw_short}
%
%\subsection{Existing systems}
%\label{rw_short_existing}
\noindent{\bf Existing systems.}
Table~\ref{tab:rw} characterizes existing systems from the point of view of the sub-properties of fast
read-only transactions that they achieve, their support for multi-object write transactions, and their target consistency
level. 
%The table shows that, c
Consistently with our theorem, none among the systems that target
the system model described in Section~\ref{sec:model-and-def}  
implements multi-object write transactions
and fast read-only transactions.  
Several systems achieve three out of the four properties we consider, 
and COPS-SNOW is the only one that implements fast read-only-transactions while complying with our system 
model. % (but it does not support multi-object write transactions). 
 Our theorem implies that the design of these systems
cannot be improved with respect to the properties we consider. 

%Table~\ref{tab:rw} shows that 
SwiftCloud and Eiger-PS
 implement fast read-only transactions and support multi-object write transactions,  %These two systems,
  %however,  
  but assume a system model that differs from the one we target.
Although they %assume that all writes eventually complete, 
 eventually complete all writes, 
the values they write may be  invisible to some clients for an indefinitely long time.
Hence, read-only transactions may see very old values of some objects, even the initial ones.
To improve the freshness of the data seen by the clients, % in these systems,
servers can communicate with clients out of the scope of transactional operations.  
This %communication between servers and clients 
 requires that the servers 
maintain a view of the connected clients.
\y{Typically,
there are far more clients than servers, so this design choice results in
reduced performance and scalability and is avoided by state-of-the-art data platforms. }
Furthermore, SwiftCloud assumes only a single partition that stores the whole data set (potentially fully replicated across multiple sites).

\ifthenelse{\boolean{picturepage}}{}{
\begin{table}[]
\centering
\scriptsize
\begin{tabular}{lccccc}
\hline
\multicolumn{1}{|c|}{\multirow{2}{*}{System}} & \multicolumn{3}{c|}{Fast ROT}                                                                  & \multicolumn{1}{c|}{\multirow{2}{*}{WTX}} & \multicolumn{1}{c|}{\multirow{2}{*}{Consistency}} \\ \cline{2-4}
\multicolumn{1}{|l|}{}                        & \multicolumn{1}{c|}{R} & \multicolumn{1}{c|}{V} & \multicolumn{1}{c|}{N} & \multicolumn{1}{l|}{}                          & \multicolumn{1}{l|}{}                             \\ \hline
RAMP~\cite{Bailis:2014}                    & $\leq 2$ & $\leq 2$ & yes         & yes                       & Read Atomicity~\cite{Bailis:2014}                           \\
COPS~\cite{lloyd_settle_2011}                    & $\leq 2$ & $\leq 2$ & yes         & no                        & Causal Consistency~\cite{ahamad_causal_1995}                           \\
Orbe~\cite{du_orbe_2013}                    & 2        & 1        & no          & no                        & Causal Consistency                           \\
GentleRain~\cite{du_gentlerain_2014}              & 2        & 1        & no          & no                        & Causal Consistency                           \\
ChainReaction~\cite{almeida_chainreaction_2013}           & $\geq 1$ & $\geq 1$ & no          & no                        & Causal Consistency                           \\
POCC~\cite{Spirovska:2017}                   & 2        & 1        & no          & no                        & Causal Consistency                           \\
Contrarian~\cite{Didona:2018}              & 2        & 1        & yes         & no                        & Causal Consistency                           \\
COPS-SNOW~\cite{lu_snow_2016}               & 1        & 1        & yes         & no                        & Causal Consistency                           \\
Eiger~\cite{lloyd_stronger_2013}                   & $\leq 3$ & $\leq 2$ & yes         & yes                       & Causal Consistency                           \\
Wren~\cite{Spirovska:2018}                    & 2        & 1        & yes         & yes                       & Causal Consistency                           \\
SwiftCloud$^\dagger$~\cite{zawirski_write_2015}             & 1        & 1        & yes         & yes                       & Causal Consistency                           \\
Cure ~\cite{akkoorath_cure_2016}                   & 2        & 1        & no          & yes                       & Causal Consistency                           \\
Yesquel~\cite{Aguilera:2015}                 & 1        & 1        & no          & yes                       & Snapshot Isolation~\cite{Berenson:1995}                           \\
Occult~\cite{mehdi_occult_2017}                  & $\geq 1$ & $\geq 1$ & yes         & yes                       & Per Client Parallel SI~\cite{mehdi_occult_2017}                         \\
Granola~\cite{Cowling:2012}                 & 2        & 1        & yes          & yes                       & Serializability~\cite{bernstein_concurrency_1987}                         \\
TAPIR~\cite{Zhang:2015}                                         & $\leq 2$                    & 1                             & yes                              & yes                                            & Serializability \\
Eiger-PS$^\dagger$~\cite{lu_snow_2016}                & 1        & 1        & yes         & yes                       & PO-Serializability~\cite{lu_snow_2016}                        \\
Spanner$^\dagger$~\cite{Corbett:2013}              & 1        & 1        & no          & yes                       & Strict Serializability~\cite{papadimitriou_serializability_1979}                         \\
DrTM~\cite{Wei:2015}                    & $\geq 1$ & $\geq 1$  & no          & yes                       & Strict Serializability                         \\
RoCoCo~\cite{Mu:2014}                  & $\geq 1$  & $\geq 1$  & no          & yes                       & Strict Serializability                         \\
RoCoCo-SNOW~\cite{lu_snow_2016}             & 1        & 1        & no          & yes                       & Strict Serializability                         \\
Calvin~\cite{calvin_thomson_2012}                                        & 2                           & 1                             & no                               & yes                                            & Strict Serializability \\
\hline
\hline
\end{tabular}
\caption{Characterization of existing systems. 
Systems with a $^\dagger$ rely on a different system model from the one we target.}\label{tab:rw}
\vspace*{-.8cm}
\end{table}
}

%\subsection{Impossibility results}
\vspace{-9pt}
~\\\noindent{\bf Impossibility results.}
Existing impossibility results on storage systems typically rely on stronger consistency or progress properties. 
Brewer \cite{brewer_conjecture_2000} conjectured the CAP theorem, according to which no implementation 
guarantees \emph{consistency}, \emph{availability}, and \emph{network partition tolerance}. 
Gilbert and Lynch \cite{gilbert_brewer_2002} formalized and proved this conjecture. 
Specifically, they formalized consistency by using the notion of \emph{atomic}
objects \cite{lamport_interprocess_1986}  (i.e., by assuming \emph{linearizability}
\cite{herlihy_linearizability_1990}, which is stronger than causal consistency). 
Roohitavaf {\em et al.} \cite{roohitavaf_causalspartan_2017} considered 
a replicated storage system implemented using \emph{data centers} (i.e., clusters of servers), 
and a model in which any value written is \emph{immediately} visible to the reads initiated in the same data center. 
They proved that it is impossible to 
ensure causal consistency, availability and network partition tolerance across data centers.
Their proof (as well as the proof of the CAP Theorem) rely on message losses,
whereas in our model no message can be lost.
        
Mahajan {\em et al.}~\cite{mahajan_consistency_2011}
proved that no implementation guarantees 
one-way convergence (a progress condition stating that if processes
communicate appropriately, then they eventually converge
on the values they read for objects), availability, and any consistency stronger than real time causal
consistency~\cite{mahajan_consistency_2011} assuming that messages may be lost. 
In their model, communication may occur among any pair of processes. 
On the contrary, in our model, communication cannot occur directly between clients,
the progress property we assume is simpler (and decoupled from the underlying communication),
and no message may be lost. 

Variants of causal consistency motivated by replicated systems 
have been presented in~\cite{attiya_limitations_2017,xiang_lower_2017}. 
Their definitions are based on the events 
that are executed at the servers (and not on the histories of operations executed in the transactions 
issued by the clients). 
Attiya {\em et al.} \cite{attiya_limitations_2017} proved that a 
(non-transactional) replicated storage system implementing 
{\em multi-valued registers} (i.e., registers for which a read 
returns the set of values written by conflicting writes) 
cannot satisfy any consistency strictly stronger than observable causal consistency. 
Xiang and Vaidya \cite{xiang_lower_2017} defined the notion of {\em replica-centric causal 
consistency}, and they proved that (non-transactional) 
replicated distributed storage systems ensuring this consistency property 
have to track writes.
These works are in different avenues than our work
and focus on other models than that in our paper.
        
Lu {\em et al.} \cite{lu_snow_2016} proved the SNOW theorem, 
which shows that no fully-functional distributed transactional system can support fast 
\emph{strictly serializable}  read-only transactions. 
%(that achieve stronger consistency than causal consistency)  are impossible. 
Lu {\em et al.} also %argued 
showed that any fully-functional distributed transactional system that achieves a consistency level 
weaker than or equal to process-ordered serializability~\cite{lu_snow_2016} %strict serializability  
 (and hence causal consistency) can support fast 
read-only transactions. \y{Tomsic et al.~\cite{tomsic_2018} further showed 
that implementing fast read-only transactions with an order-preserving consistency level 
(as is the case for causal consistency) is possible only by allowing read-only transactions 
to read possibly stale values of the objects being accessed.} These results may seem at odds with our impossibility result. 
However, these results rest on very weak assumptions on the progress guarantees of write operations. 
Although they assume that all writes eventually complete, 
the values they write may be invisible to clients for an indefinitely 
long time. %\footnote{To reduce this time, the system model allows servers to communicate with clients out of the scope of transactional operations, which is not allowed in our model (Section~\ref{sec:model-and-def}).}.
 Such a weak assumption allows the design of trivial algorithms 
in which read-only transactions can return arbitrarily old values --even the initial ones-- for the objects they read.

Recently, Didona {\em et. al}~\cite{Didona:2018} showed a lower bound
on the number of bits that must be communicated in order to support
fast causally consistent read-only transactions in distributed storage systems.
 % that may reduce the throughput achievable by a system. 
%That work   discusses the {\em performance} implications of achieving   fast read-only transactions. 
%The employed model in that work does not consider multi-object write transactions.
On the contrary, this paper focuses on the {\em design} implications of fast read transactions
in distributed transactional such systems, 
showing that they are incompatible with multi-object write transactions.

Since the introduction of causal consistency by Ahamad {\em et al.} \cite{ahamad_causal_1995} 
for a shared memory system,  
other versions of causal consistency has been studied~\cite{raynal_from_1997,lloyd_settle_2011}. 
Our result holds if we replace our definition of causal consistency
with those provided in these papers.

\section{Conclusion}
We present an impossibility result that establishes a fundamental trade-off in the design of
 distributed transactional storage systems: fast read transactions cannot be achieved by
  fully-functional transactional storage systems. The design of such systems must either
   sacrifice fast read transactions, or must settle for reduced functionality, i.e., support
    only single-object write transactions.

Unlike most previous work on distributed transactional systems, which target strong consistency, 
our result assumes only causal consistency. This broadens the scope of our result 
which applies also to systems that implement any consistency level stronger
 than causal consistency, or a hybrid consistency level that includes causal consistency. 
 Proving our result under such weak consistency model is nontrivial and required us to devise 
a complex proof.

Our result sheds light on the design choices of state-of-the-art distributed transactional
storage systems, and is useful for the architects of such systems because it identifies impossible designs. 
%At the same time, 

Our result also opens several interesting research questions, such as
investigating which  is the weakest consistency condition for which our impossibility result holds.
In addition, it is interesting to further investigate the design of systems that provide some of the combinations
of the studied properties, as discussed in Section~\ref{sec:limits}.

\section*{Acknowledgements}
This research has been supported by the European Union's H2020-FETHPC-2016 
programme under grant agreement No 754337 (project EuroExa), by the European ERC Grant 339539-AOC, and by an EcoCloud post-doctoral research fellowship.  Part of this work has been performed while P. Fatourou was working as a visiting professor at EPFL.


\clearpage

\appendix

\section{The General Impossibility Result}
\label{sec:gene-imp}
\begin{table*}[h!]
\scriptsize
\begin{tabular}{lll}
\hline
\multicolumn{1}{|l|}{\textbf{Symbol}} & \multicolumn{1}{l|}{\textbf{Meaning}}                                                                                                                                                                                                         & \multicolumn{1}{l|}{\textbf{Used in}}                                                                                           \\ \hline
\multicolumn{1}{|l|}{$X_i$}           & \multicolumn{1}{l|}{object $i$}                                                                                                                                                                                                               & \multicolumn{1}{l|}{System model and proofs}                                                                                    \\ \hline
\multicolumn{1}{|l|}{$x_i^{in}$}      & \multicolumn{1}{l|}{Initial, old value of object $x_i$}                                                                                                                                                                                       & \multicolumn{1}{l|}{System model and proofs}                                                                                    \\ \hline
\multicolumn{1}{|l|}{$p_i$}           & \multicolumn{1}{l|}{Server that stores $x_i$}                                                                                                                                                                                                 & \multicolumn{1}{l|}{System model and proofs}                                                                                    \\ \hline
\multicolumn{1}{|l|}{$T_i^{in}$}      & \multicolumn{1}{l|}{Transaction that writes $x_i^{in}$}                                                                                                                                                                                       & \multicolumn{1}{l|}{System model and proofs}                                                                                    \\ \hline
\multicolumn{1}{|l|}{$c_i^{in}$}      & \multicolumn{1}{l|}{Client that performs $T_i^{in}$}                                                                                                                                                                                          & \multicolumn{1}{l|}{System model and proofs}                                                                                    \\ \hline
\multicolumn{1}{|l|}{$c_w$}           & \multicolumn{1}{l|}{\begin{tabular}[c]{@{}l@{}}Client that reads $x_0^{in}$ and $x_1^{in}$ and then performs\\ a write transaction that writes new values $x_0$, $x_1$\end{tabular}}                                                          & \multicolumn{1}{l|}{System model and proofs}                                                                                    \\ \hline
\multicolumn{1}{|l|}{$T_r^{in}$}      & \multicolumn{1}{l|}{Read-only transaction issued by $c_w$ to read $X_0$, $X_1$}                                                                                                                                                               & \multicolumn{1}{l|}{System model and proofs}                                                                                    \\ \hline
\multicolumn{1}{|l|}{$T_w$}           & \multicolumn{1}{l|}{Write-only transaction issued by $c_w$ that writes new values for $X_0$, $X_1$}                                                                                                                                           & \multicolumn{1}{l|}{All proofs}                                                                                                 \\ \hline
\multicolumn{1}{|l|}{$x_i$}           & \multicolumn{1}{l|}{New value of $X_i$, written by $T_w$}                                                                                                                                                                                     & \multicolumn{1}{l|}{All proofs}                                                                                                 \\ \hline
\multicolumn{1}{|l|}{$c_r$}           & \multicolumn{1}{l|}{Client that invokes $T_r$ to read $X_0$, $Y_0$}                                                                                                                                                                           & \multicolumn{1}{l|}{System model and proofs}                                                                                    \\ \hline
\multicolumn{1}{|l|}{$m_{ji}$}        & \multicolumn{1}{l|}{j-th message sent by $p_i$}                                                                                                                                                                                               & \multicolumn{1}{l|}{System model and proofs}                                                                                    \\ \hline
\multicolumn{1}{|l|}{$Q_{in}$}        & \multicolumn{1}{l|}{Initial configuration}                                                                                                                                                                                                    & \multicolumn{1}{l|}{Proof of Lemma 1 (See Figure 1)}                                                                            \\ \hline
\multicolumn{1}{|l|}{$Q_{0}$}         & \multicolumn{1}{l|}{Configuration in which $x_0^{in}$, $x_1^{in}$ become visible}                                                                                                                                                             & \multicolumn{1}{l|}{Proof of Lemma 1(See Figure 1)}                                                                             \\ \hline
\multicolumn{1}{|l|}{$C_{0}$}         & \multicolumn{1}{l|}{Configuration in which $T_r^{in}$ has returned $x_0^{in}$, $x_1^{in}$ to $c_w$.}                                                                                                                                          & \multicolumn{1}{l|}{Proof of Lemma 1(See Figure 1)}                                                                             \\ \hline
\multicolumn{1}{|l|}{$T_r$}           & \multicolumn{1}{l|}{Read-only transaction issued by $c_r$}                                                                                                                                                                                    & \multicolumn{1}{l|}{All proofs}                                                                                                 \\ \hline
\multicolumn{1}{|l|}{$\gamma_{old}$}  & \multicolumn{1}{l|}{Execution in which $c_r$ reads $x_0^{in}$, $x_1^{in}$}                                                                                                                                                                    & \multicolumn{1}{l|}{\begin{tabular}[c]{@{}l@{}}Construction 1 (See Figure 2),\\ Proof of Lemma 3\end{tabular}}                  \\ \hline
\multicolumn{1}{|l|}{$\gamma_{new}$}  & \multicolumn{1}{l|}{Execution in which $c_r$ reads $x_0$, $x_1$}                                                                                                                                                                              & \multicolumn{1}{l|}{\begin{tabular}[c]{@{}l@{}}Construction 2  (See Figure 2) \\ Proof of Lemma 3\end{tabular}}                 \\ \hline
\multicolumn{1}{|l|}{$\gamma$}        & \multicolumn{1}{l|}{\begin{tabular}[c]{@{}l@{}}Contradictory execution in which $c_r$ reads a mix of old and new values for $X_0$, $X_1$.\\  $\gamma$ results from the mix of sub-executions of $\gamma_{old}$ and $\gamma_{new}$\end{tabular}} & \multicolumn{1}{l|}{Proof of claim 1 of Lemma 3 (see Figure 3)}                                                                 \\ \hline
\multicolumn{1}{|l|}{$\beta$}         & \multicolumn{1}{l|}{Sub-sequence of $\gamma$}                                                                                                                                                                                                 & \multicolumn{1}{l|}{Proof of claim 1 of Lemma 3 (see Figure 3)}                                                                 \\ \hline
\multicolumn{1}{|l|}{$\sigma_{old}$}  & \multicolumn{1}{l|}{Sub-execution of $\gamma_{old}$ that leads $c_r$ to read $x_0^{in}$}                                                                                                                                                      & \multicolumn{1}{l|}{\begin{tabular}[c]{@{}l@{}}Construction 1 (see Figure 2),\\ Proof of claim 1 and 2 of Lemma 3\end{tabular}} \\ \hline
\multicolumn{1}{|l|}{$\sigma_{new}$}  & \multicolumn{1}{l|}{Sub-execution of $\gamma_{new}$ that leads $T_r$ of $c_r$ to read new values}                                                                                                                                             & \multicolumn{1}{l|}{\begin{tabular}[c]{@{}l@{}}Construction 1 (see Figure 2),\\ Proof of claim 1 and 2 of Lemma 3\end{tabular}}  \\ \hline
\multicolumn{1}{|l|}{$\beta_{new}$}   & \multicolumn{1}{l|}{Sequence of steps that brings the system to a configuration in which $x_0$, $x_1$ are visible}                                                                                                                            & \multicolumn{1}{l|}{Proof of claim 1 and 2 of Lemma 3 (see Figure 3)}                                                           \\ \hline
\multicolumn{1}{|l|}{$\delta$}        & \multicolumn{1}{l|}{\begin{tabular}[c]{@{}l@{}}Contradictory execution in which $c_r$ reads a mix of old and new values for $X_0$, $X_1$. \\ $\delta$ results from the mix of sub-executions of $\gamma_{old}$ and $\gamma_{new}$\end{tabular}} & \multicolumn{1}{l|}{Proof of claim 2 of Lemma 3}                                                                                \\ \hline
\multicolumn{1}{|l|}{$\rho$}          & \multicolumn{1}{l|}{Subsequence of $\delta$ (similar to $\beta$)}                                                                                                                                                                             & \multicolumn{1}{l|}{Proof of claim 2 of Lemma 3}                                                                                \\ \hline
\multicolumn{1}{|l|}{$\alpha$}        & \multicolumn{1}{l|}{\begin{tabular}[c]{@{}l@{}}Infinite execution that contains just one write-only transaction $T_w$\\  and such that the values written by $T_w$ never become visible\end{tabular}}                                         & \multicolumn{1}{l|}{Proof of Lemma 3}                                                                                           \\ \hline
\multicolumn{1}{|l|}{$\alpha_k$}      & \multicolumn{1}{l|}{Prefix of $\alpha$}                                                                                                                                                                                                       & \multicolumn{1}{l|}{Proof of Lemma 3}                                                                                           \\ \hline
                                      &  
\end{tabular}
\caption{Table of symbols.}
\label{table:symbols}
\end{table*}

We prove our impossibility result i.e., 
if a causally consistent implementation
of a transactional storage system supports write-only transactions that write
to more than one object, then it cannot also provide 
fast read-only transactions for the general case, where the system
has any number of servers and the system is partially replicated. 
By partial replication, we mean that each server stores a different set of objects 
but these sets are not disjoint.
In addition, no server stores all objects; i.e., for any server, there is an object such that the server does not store it.
To access an object, the algorithm is allowed to access all servers that store the object.
In this case, we need to revisit our definition of a fast read-only transaction.

\begin{definition}[General fast read-only transaction]
\label{def:gene-fast-op}
We say that an implementation of a distributed storage system supports 
fast read-only transactions, if in each execution $\alpha$ it produces, 
the following hold for every read-only transaction executed in $\alpha$: 
\begin{enumerate}
\item
{\bf Non-blocking and One-Roundtrip Property.} Same as in Definition \ref{def:fast-op};
\item
{\bf General one-value messages.}
(a) Each message sent from a server to a client does not contain any value 
that has been written by some write transaction in $\alpha$ to an object $X$ if $X$ is not stored on the server;
(b)
Let $\Sigma_X$ be the set of servers which store an object $X$. 
For each object $X$ read by a client, only one server in $\Sigma_X$ sends a message to the client and the message contains only one value that has been written by some write transaction in $\alpha$ to $X$.
%that is stored in the server and is read by the client
% \footnote{We remark that the message may also contain
%some metadata (e.g. a timestamp), as long as these metadata 
%do not reveal any information about other transactions to the client.}
\end{enumerate}
\end{definition}

\begin{theorem}
\label{thm:gene-imp}
No causally consistent implementation of a transactional storage system 
that supports transactions which can concurrently read and write multiple objects, 
provides fast read-only transactions.
\end{theorem}

(The claim of Theorem \ref{thm:gene-imp} remains the same as that of Theorem \ref{thm:imp}. The only difference is that we now assume a partially replicated system.)

The proof is by contradiction.
Assume that there exists a causally consistent implementation $\Pi$ 
which supports concurrent read-write transactions that access multiple objects and provides fast read-only transactions.
We construct an infinite execution $\alpha$ which contains just a write-only transaction $T_w$, invoked by some client $c_w$,
and we show that the values written by $T_w$ never become visible. 

Let $N+1$ be the number of objects stored in the system. 
For convenience, every execution we consider in this paper, starts with 
the execution of $N+1$ initial transactions, starting from $Q_{in}$. 
Each transaction is denoted by $T_i^{in}$ (invoked by $c_i^{in}$) that writes some initial value $x_i^{in}$ in $X_i$.
We denote by $QE_0$ a reachable configuration in which all values
 $x_0^{in}, x_1^{in}, \ldots, x_N^{in}$ are visible and all buffers are empty.\footnote{If such reachable configuration does not exist, it is already shown that for $x_0^{in}, x_1^{in}, \ldots, x_N^{in}$ to be visible, there are always messages in buffers, which contradicts the progress condition.}
We also revisit the definition of value visibility to exclude from later executions, these clients $c_0^{in}, c_1^{in}, $ $\ldots, $ $c_N^{in}$ whose sole purpose is to initialize values for each object.

\begin{definition}[General value visibility]
\label{def:gene-vis-val}
Consider any object $X$ and let $C$ be any reachable configuration
which is either quiescent or just a write-only transaction (by a client $c_w$)
writing a value $x$ into $X$ is active in $C$. 
Value $x$ is {\em visible} in $C$, if and only if the following holds: in every legal execution 
starting from $C$ which contains just a read-only transaction $T_r$ 
\y{(invoked by any client  $c \not\in \{c_w, c_0^{in}, c_1^{in}, \ldots, c_N^{in} \}$)}
that reads $X$, $x$ is returned as the value of $X$ for $T_r$.
\end{definition}

For convenience, every execution we construct hereafter, starts with 
the execution of a read-only transaction $T_r^{in} = (r(X_0)*, r(X_1)*,$ $ \ldots, r(X_N)*)$ by a client $c_w, c_w\notin \{c_0^{in}, c_1^{in}, \ldots, c_N^{in}\}$ (starting from $QE_0$). Since $T_r^{in}$ is a fast read-only transaction, $T_r^{in}$ completes; since $x_0^{in}, x_1^{in}, \ldots, x_N^{in}$ are visible in $QE_0$, $c_w$ returns $(x_0^{in}, x_1^{in}, \ldots, x_N^{in})$ for $T_r^{in}$.
We denote by $CE_0$ a reachable configuration starting from the configuration in which $T_r^{in}$ completes, in which all buffers are empty (i.e., no message is in transit).

For every reachable configuration $C$ (starting from $CE_0$) and every legal execution $\gamma$ from $C$,
we denote by $RC(C,\gamma)$ the configuration that results from the execution
of $\gamma$ starting from $C$.
Given two executions $\gamma_1$ and $\gamma_2$, we denote by $\gamma_1 \cdot \gamma_2$
the concatenation of $\gamma_1$ with $\gamma_2$, i.e., $\gamma_1 \cdot \gamma_2$ is an execution
consisting of all events of $\gamma_1$ followed by all events of $\gamma_2$ (in order).

\subsection{Preliminary Observations and Lemmas}

We start with some useful observations. The first is an immediate consequence of the fact
that $\Pi$ ensures causal consistency. 
\begin{observation}
\label{ob:gene-granularity}
Let $\gamma$ be any legal execution of $\Pi$ starting from $CE_0$ which contains two transactions: 
client $c_w$ invokes a write-only transaction
$T_w = (w(X_0)x_0, w(X_1)x_1, \ldots, w(X_N)x_N)$, and a different client $c_r$ invokes a read-only transaction $T_r = (r(X_0)*, r(X_1)*, \ldots, $ $r(X_N)*)$ which completes in $\gamma$.
Let $v_i$ be the value which $c_r$ returns for $T_r$, i.e., $T_r = (r(X_0)v_0, r(X_1)v_1, \ldots, r(X_N)v_N)$.
Then, either $v_i = x_i, \forall i\in\{0,1,\ldots,N\}$, or $v_i = x_i^{in} \forall i\in\{0,1,\ldots,N\}$.
\end{observation}

\begin{observation}
\label{ob:gene-not-visible}
Let $\tau$ be any arbitrary legal execution starting from $CE_0$ which contains just one transaction: 
client $c_w$ executes a write-only transaction $T_w = (w(X_0)x_0, w(X_1)x_1, \ldots, w(X_N)x_N$ 
(i.e., $c_w$ and the servers take steps to execute $T_w$).
Let $C$ be any reachable configuration when $\tau$ is applied from $CE_0$.
If $\exists i, x_i$ is not visible in $C$, then
there exists at least one client $c_r \notin \{c_w, c_0^{in}, c_1^{in}, \ldots, c_N^{in} \}$
such that if, starting from $C$, $c_r$ executes a read-only transaction 
$T_r = (r(X_0)*, $ $r(X_1)*, \ldots, r(X_N)*)$, then $c_r$ returns 
$(x_0^{in}, x_1^{in}, \ldots, x_N^{in})$ for $T_r$.
\end{observation}

\begin{proof}
To derive a contradiction, 
assume that for any client 
$c_r \not\in \{c_w, c_0^{in}, c_1^{in},\ldots, c_N^{in} \}$
$c_r$ invokes $T_r$ and $c_r$ does not return 
$(x_0^{in}, x_1^{in}, \ldots, x_N^{in})$. 
Notice that by convention, as $\tau$ starts from $CE_0$, $v_i = \bot$ for any $i\in\{0,1, \ldots, N\}$ is not an option.
(Notice also that since $\Pi$ ensures that read-only transactions are fast, 
$T_r$ completes.)
Then by Observation \ref{ob:gene-granularity}, $c_r$ returns $(x_0, x_1, \ldots, x_N)$. Since $c_r$ is chosen to be an arbitrary client 
not in $\{c_w, c_0^{in}, c_1^{in}, \ldots, c_N^{in}\}$, 
Definition \ref{def:gene-vis-val} 
implies that at $C$, $x_i, \forall i\in\{0,1,\ldots,N\}$ is visible.  
This contradicts the hypothesis that at least one of the values is not visible.
\end{proof}

Notice that if $C = CE_0$, then Observation~\ref{ob:gene-not-visible} holds for all clients
not in $\{c_w, c_0^{in}, c_1^{in}, \ldots, c_N^{in} \}$.

In the rest of the proof, we will repeatedly employ the two executions we present 
in Constructions~\ref{gene-gold_construction} and~\ref{gene-gnew_construction} (each time using different parameters for the construction).

\begin{construction}[Construction of execution $\gamma_{old}(C, p, c_r)$ and execution $\sigma_{old}(C, p, c_r)$]
\label{gene-gold_construction}
Let $\tau$ be any arbitrary legal execution starting from $CE_0$ which contains just one transaction: 
%a 
client $c_w$ executes a write-only transaction $T_w = (w(X_0)x_0, w(X_1)x_1, \ldots, w(X_N)x_N)$. 
take steps to execute $T_w$).
Fix any server $p$. 
For every reachable configuration $C$ when $\tau$ is applied from $CE_0$
in which the following holds: 
{\em no value other than $x_i^{in}$ is visible for $X_i$} for at least one $i\in\{0,1,\ldots,N\}$,
and for every client $c_r$ that satisfies Observation~\ref{ob:gene-not-visible},
we define execution $\gamma_{old}(C, p, c_r)$ as follows. 
In $\gamma_{old}(C, p, c_r)$, first $c_r$ invokes $T_r$ starting from $C$. 
So, $c_r$ takes steps and since it reads all objects,  
$c_r$ sends a message $mo_q(C, p, c_r)$ to every server $q$.
Next, $mo_q(C, p, c_r)$ is delivered for every server $q, q\neq p$ and then each server $q, q\neq p$ takes a step and receives $mo_q(C, p, c_r)$.
Since $T_r$ is a fast transaction, once $q$ receives $mo_q(C, p, c_r)$, 
either $q$ sends a response $mo_q'(C, p, c_r)$ to $c_r$ within the same step, or $q$ does not.
Denote by $\sigma_{old}(C, p, c_r)$ this sequence of steps starting from $C$ to the step in which every server $q, q\neq p$ has taken a step and sent the response (if $q$ sends it). 
Next,
we let the remaining message, i.e., $mo_p(C, p, c_r)$, be delivered, and let 
$p$ take a step to receive $mo_p(C, p, c_r)$.
By one-roundtrip property, once $p$ receives $mo_p(C, p, c_r)$, either $p$ sends a response to $c_r$ within the same step, or $p$ does not.
Finally, we let $c_r$ take steps, receive the responses from $p$ and every server $q, q\neq p$ (if they send a response), 
complete $T_r$, and return a response for it.
\end{construction}

By the way $\gamma_{old}(C, p, c_r)$ is constructed and by 
Observation \ref{ob:gene-not-visible}, we get the following.

\begin{observation}
\label{obs:gene-gold}
The following claims hold:
\begin{enumerate}
\item \label{gene-gold_legal} Execution $\gamma_{old}(C, p, c_r)$ is legal from $C$, and 
\item \label{gene-gold_return} The return value for $T_r$ in $\gamma_{old}(C, p, c_r)$ is  
$(x_0^{in}, x_1^{in}, \ldots, x_N^{in})$.
\end{enumerate}
\end{observation}

\begin{construction}[Construction of execution $\gamma_{new}(C, p, c_r)$ and execution $\sigma_{new}(C, p, c_r)$]
\label{gene-gnew_construction}
Let $\tau$ be any arbitrary legal execution starting from $CE_0$ which contains just one transaction: 
client $c_w$ executes a write-only transaction $T_w = (w(X_0)x_0, w(X_1)x_1, \ldots, w(X_N)x_N)$.
Fix any server $p$ and let $c_r$ be any client not in $\{c_w, c_0^{in}, c_1^{in}, \ldots, c_N^{in}\}$.
For every reachable configuration $C$ when $\tau$ is applied from $CE_0$,
in which 
{\em $x_i$ is visible} for $X_i$ for at least one $i\in\{0,1,\ldots, N\}$, 
we define execution $\gamma_{new}(C, p, c_r)$ as follows. 
In $\gamma_{new}$, starting from $C$, $c_r$ invokes $T_r$ and takes steps until  
it sends a message $mn_q(C, p, c_r)$ to every server $q$.
Let $C_{new}(C, p, c_r)$ be the resulting configuration.
Next, we let $mn_p(C, p, c_r)$ be delivered.
Then, $p$ takes a step and receives $mn_p(C, p, c_r)$.
Since $T_r$ is a fast transaction, once 
$p$ receives $mn_p(C, p, c_r)$, 
either $p$ sends a response $mn_p'(C, p, c_r)$ to $c_r$ within the same step, or $p$ does not.
This sequence of steps starting from $C_{new}(C, p, c_r)$ to the step which $p$ takes (inclusive) is denoted by $\sigma_{new}(C, p, c_r)$. 
Next, 
we let $mn_q(C, p, c_r)$ be delivered for any server $q, q\neq p$, and 
let each $q$ take a step to receive $mn_q(C, p, c_r)$; 
once $q$ receives $mn_q(C, p, c_r)$, either $q$ sends a response $mn_q'(C, p, c_r)$ to $c_r$ within the same step, or $q$ does not.
Finally, we let $c_r$ take steps. 
Since $T_r$ is a fast transaction, once $c_r$ receives the responses from $p$ and every server $q, q\neq p$ (if they send a response to $c_r$), $T_r$ completes
and $c_r$ returns a response for $T_r$.
\end{construction}

By the way $\gamma_{new}(C, p, c_r)$ is constructed and by 
Definition \ref{def:gene-vis-val}, we get the following.

\begin{observation}
\label{obs:gene-gnew}
The following claims hold for $\gamma_{new}(C, p, c_r)$:
\begin{enumerate}
\item \label{gene-gnew_legal} Execution $\gamma_{new}(C, p, c_r)$ is legal from $C$, and 
\item \label{gene-gnew_return} The return value for $T_r$ in $\gamma_{new}(C, p, c_r)$ is  $(x_0, x_1, \ldots, x_N)$. 
\end{enumerate}
\end{observation}

We are now ready to construct an execution $\alpha_1$ (that will be a prefix of $\alpha$) 
in which for at least one server $q$, $q$ sends at least one message.

\begin{lemma}
\label{lma:gene-base}
In any arbitrary legal execution starting from $CE_0$ which contains just one transaction: 
client $c_w$ executes a write-only transaction $T_w = (w(X_0)x_0, w(X_1)x_1, \ldots, w(X_N)x_N)$, 
it holds that for at least one server $q$ in the system,
at least one of the following two happens:
\begin{itemize}
\item $q$ sends a message to a server other than $q$ (i.e., $q$ sends an {\em explicit} message to a server);
\item $q$ sends a message to $c_w$ so that after $c_w$ receives this message, $c_w$ sends a message to a server other than $q$
(i.e., $q$ sends an {\em implicit} message to a server).
\end{itemize}
\end{lemma}

\begin{proof}
To derive a contradiction, 
assume that for every server $q$ in the system, the following hold:
\begin{itemize}
\item $q$ sends no message to any other server;
\item $q$ sends no message to $c_w$ so that after $c_w$ receives the message, $c_w$ sends a message to a server other than $q$.
\end{itemize}

To derive a contradiction, we will construct an execution $\gamma$, where in addition to transaction $T_w$,
a read-only transaction $T_r=(r(X_0)*, r(X_1)*, \ldots, r(X_N)*)$ is invoked by a client $c_r$ 
not in $\{c_w, c_0^{in}, $ $c_1^{in}, \ldots, c_N^{in}\}$
and $T_r$ completes in $\gamma$. We will show that $\gamma$ contradicts Observation \ref{ob:gene-granularity}.

To construct $\gamma$, we first define an execution $\beta$ and some subsequences of it,
and we study their properties. 
Starting from $CE_0$, $c_w$ invokes $T_w$ and executes solo
(i.e., only $c_w$ and all servers take steps) until $x_i$ is visible for any $i\in\{0,1,\ldots,N\}$ 
(minimal progress for write-only transactions implies that there is a configuration 
in which $x_i$ is indeed visible for any $i\in\{0,1,\ldots,N\}$).
Let $C_v$ be the first configuration where $x_i$ is visible for any $i\in\{0,1,\ldots,N\}$.
Let $\beta$ be the sequence of steps which are taken from $CE_0$ to $C_v$. 
For any server $p$, consider execution $\gamma_{new}(C_v, p, c_r)$.
Because in $\gamma_{new}(C_v, p, c_r)$, starting from $C_v$, only the messages sent by $c_w$  are delivered before $c_r$ takes the steps to complete $T_r$, thus to any server $p$, $\gamma_{new}(C_v, p, c_r)$ and $\gamma_{new}(C_v, s, c_r)$ are indistinguishable for any server $s, s\neq p$.
Thus when $c_r$ takes the steps to complete $T_r$, $c_r$ receives the same set of responses in $\gamma_{new}(C_v, p, c_r)$ for any server $p$.
In order to exclude empty and garbage responses, let $M$ be the smallest set of non-empty responses where each response contains a value that has been written by some write transaction. As the system is partially replicated, the size of $M$ is at least two.
W.l.o.g., we fix any server $p$ which sends a response included in $M$ and thus fix execution $\gamma_{new}(C_v, p, c_r)$.
Then, Observation~\ref{obs:gene-gnew} implies that the response for $T_r$ in $\gamma_{new}(C_v, p, c_r)$
is $(x_0, x_1, \ldots, x_N)$.
For simplicity, we omit 
$(C_v, p, c_r)$
from the notation in the rest of the proof of this lemma.

Recall that $\beta$ is the sequence of steps which are taken from $CE_0$ to $C_v$. 
Let $\beta'_p$ be the shortest prefix of $\beta$ which contains all messages sent by $c_w$ to $p$.
Let $\beta_p$ be the subsequence of $\beta'_p$ in which all steps taken by any server $q, q\neq p$ have been removed.
Let $\beta'_s$ be the suffix of $\beta$ after $\beta_p'$.
Let $\beta_s$ be the subsequence of $\beta'_s$ containing only steps by $p$.
Let $\beta_{new}$ be $\beta_p\cdot \beta_s$.
We now argue that $\beta_{new}$ is legal from $CE_{0}$.
Because $\beta'_p$ is a prefix of $\beta$, $\beta'_p$ is legal from $CE_{0}$.
By assumption, $p$ does not receive any message 
from any server $q, q\neq p$ in $\beta'_p$. Moreover,  
any server $q, q\neq p$ sends no message to $c_w$ so that after $c_w$ receives the message, $q$ sends a message to $p$.
By definition, $\beta'_p$ ends with a message sent by $c_w$ to $p$ (if it is not empty). It follows that 
$c_w$ does not receive any message from $q$ in $\beta'_p$.
Thus, $\beta_p$ is legal from $CE_0$. Moreover,
$RC(CE_0, \beta'_p)$ and $RC(CE_0, \beta_p)$ are indistinguishable to $p$ and $c_w$.

We next argue that $\beta_s$ is legal from $RC(CE_0, \beta_p')$. 
We notice that in $\beta_s$, only $p$ takes steps.
Then because $RC(CE_0, \beta'_p)$ and $RC(CE_0, \beta_p)$ are indistinguishable to $p$,  
by arguing that $\beta_s$ is legal from $RC(CE_0, \beta_p')$,
we also argue that $\beta_s$ is legal from $RC(CE_0, \beta_p)$.
Below we argue that $\beta_s$ is indeed legal from $RC(CE_0, \beta_p')$. 
By assumption, any server $q, q\neq p$ does not send any message to $p$ in $\beta$,
so $p$ does not receive any message from any server $q, q\neq p$ in $\beta_s$.
By definition of $\beta_s'$, $c_w$ does not send any message to $p$ in $\beta'_s$, 
(because all messages from $c_w$ to $p$ are sent in $\beta_p'$),
so any message $p$ receives from $c_w$ in $\beta_s$ has been sent in $\beta_p'$. 
Thus $\beta_s$ is legal from $RC(CE_0, \beta_p')$,
 and $\beta_s$ is legal from $RC(CE_0, \beta_p)$. 
Moreover, 
by arguing that $\beta_s$ is legal from $RC(CE_0, \beta_p')$,
we also show that $RC(CE_0, \beta_p'\cdot\beta_s)$ and $C_v = RC(CE_0, \beta_p'\cdot\beta_s')$ are indistinguishable to $p$. 
Because $RC(CE_0, \beta'_p)$ and $RC(CE_0, \beta_p)$ are indistinguishable to $p$,
it follows that 
$RC(CE_0, \beta_{new})$ and $C_v$ are indistinguishable to $p$.

To construct $\gamma$, we have to utilize executions $\sigma_{old}(CE_0, p, c_r)$
and $\sigma_{new}(C_v, p, c_r)$ and refer to configuration $C_{new}(C_v, p, c_r)$.
(For simplicity, we omit $(CE_0, p, c_r)$
and $(C_v, p, c_r)$ from the notations below.)
Specifically, starting from $CE_0$, 
$\sigma_{old}$ is first applied.
Recall that for any server $q, q\neq p$, in its last step of $\sigma_{old}$, $q$ sends at most one message $mo_q'$ to $c_r$.
Since $CE_0$ and $RC(CE_0, \sigma_{old})$ are indistinguishable to $c_w$ and $p$, and since only $c_w$ and $p$ 
take steps in $\beta_{new}$, $\beta_{new}$ is legal from $RC(CE_0, \sigma_{old})$. So, to construct $\gamma$,
$\beta_{new}$ is applied from $RC(CE_0, \sigma_{old})$.
The resulting configuration is $RC(CE_0, \sigma_{old} \cdot \beta_{new})$.
Since the processes which take steps in $\sigma_{old}$ and $\beta_{new}$ are disjoint,
then $RC(CE_0, \sigma_{old} \cdot \beta_{new})$ and $RC(CE_0,  \beta_{new}\cdot \sigma_{old} )$ are indistinguishable.
Recall that to $p$, $RC(CE_0, \beta_{new})$ and $C_v$ are indistinguishable.
Therefore, to $p$, 
$RC(CE_0,  \beta_{new}\cdot \sigma_{old} )$ and $C_{new}$
are indistinguishable.
Because only $p$ takes steps in $\sigma_{new}$,
$\sigma_{new}$ is therefore legal from 
$RC(CE_0,  \beta_{new}\cdot \sigma_{old} )$
and also legal from $RC(CE_0,  \sigma_{old}\cdot\beta_{new})$.
We apply it from $RC(CE_0,  \sigma_{old}\cdot\beta_{new})$.
Recall that in the last step of $\sigma_{new}$, $p$ sends a message $mn_p'$ to $c_r$.
Finally, $mo_q'$ for any server $q, q\neq p$ (if $q$ sends a response) and $mn_p'$ are delivered in $\gamma$, and next, 
$c_r$ takes steps. Since $T_r$ is a fast transaction, $T_r$ completes. 
This concludes the construction of $\gamma$. 

In $\gamma$, $c_r$ executes only transaction $T_r$.
Thus, $c_r$ decides the response for $T_r$ based solely on  the values included in $mo_q'$ for any server $q, q\neq p$ (if $q$ sends a response) and $mn_p'$.
Since $mo_q'$ for any server $q, q\neq p$ is sent before $c_w$ takes steps, $mo_q'$ contains no value in $\{x_0, x_1, \ldots, x_N\}$.
Recall that $mn_p'$ is sent in $\gamma_{new}$.
By the one-value messages property, $mn_p'$ contains only values in $\{x_0, x_1, \ldots,x_N\}$.
As the system is partially replicated, $mn_p'$ cannot contain all values in $\{x_0, x_1, \ldots, x_N\}$.
It follows that in $\gamma$, the return value of $c_r$ for $T_r$ 
can be neither $(x_0, x_1, \ldots, x_N)$ nor $(x_0^{in}, x_1^{in}, \ldots, x_N^{in})$.
This contradicts Observation \ref{ob:gene-granularity}.
\end{proof}

By Lemma \ref{lma:gene-base}, we have that
for at least one server $q$ in the system, at least one of the following two happens 
in any arbitrary legal execution starting from $CE_0$ which contains 
just one transaction $T_w = (w(X_0)x_0, w(X_1)x_1, \ldots, w(X_N)x_N)$ executed by $c_w$:
(1) $q$ sends a message to a server other than $q$;
(2) $q$ sends a message to $c_w$ so that: after $c_w$ receives this message, $c_w$ sends a message to a server other than $q$.
Let $m_1$ be the first message that satisfies any of the two claims above.

We are now ready to construct execution $\alpha_1$. 
In $\alpha_1$, $T_w$ executes solo starting from $CE_0$ until $m_1$ is sent.
We denote by $C_1$ the configuration that results when $\alpha_1$ is applied from $CE_0$.
%We prove that: 

\begin{lemma}
\label{lma:gene-e0}
At $C_1$, $x_0, x_1, \ldots, x_N$ are not visible.
\end{lemma}

\begin{proof}
Assume that 
at $C_1$,
$x_i$ is visible for some $i\in\{0,1, 2,\ldots, N\}$.
To derive a contradiction, we will construct an execution $\gamma$ where in addition to $T_w$, 
a read-only transaction $T_r=(r(X_0)*, r(X_1)*, \ldots, $ $r(X_N)*)$ is invoked 
by a client $c_r \not\in \{ c_w, c_0^{in}, c_1^{in}, \ldots, c_N^{in}\}$,
and $c_r$ returns from $T_r$. 
We will show that $\gamma$ contradicts Observation \ref{ob:gene-granularity}.

We fix the server $p$ in a similar way as in the proof of Lemma \ref{lma:gene-base} given execution $\gamma_{new}(C_1, s, c_r)$ for every server $s$.
Then we construct $\gamma$ by utilizing $\gamma_{old}(CE_0, p, c_r)$
and $\gamma_{new}(C_1, p, c_r)$. 
By Observation~\ref{obs:gene-gold}, $c_r$ returns
$(x_0^{in}, x_1^{in}, \ldots, x_N^{in})$
for $T_r$ in $\gamma_{old}(CE_0, p, c_r)$.
Because by assumption, $x_i$ is visible at $C_1$, Observation~\ref{obs:gene-gnew} implies that 
$c_r$ returns $(x_0, x_1, \ldots, x_N)$ for $T_r$ in $\gamma_{new}(C_1, p, c_r)$.

To construct $\gamma$, we first define an execution $\beta$ and some subsequences of it,
and we study their properties. 
Let $\beta$ be the sequence of steps which are taken from $CE_0$ to $C_1$.
Let $\beta'_p$ be the shortest prefix of $\beta$ which contains all messages sent by $c_w$ to $p$.
Let $\beta_p$ be the subsequence of $\beta'_p$ in which all steps taken by any server $q, q\neq p$ have been removed.
Let $\beta'_s$ be the suffix of $\beta$ after $\beta_p'$.
Let $\beta_s$ be the subsequence of $\beta'_s$ containing only steps by $p$.
Let $\beta_{new}$ be $\beta_p\cdot \beta_s$.
We now argue that $\beta_{new}$ is legal from $CE_{0}$.
Because $\beta'_p$ is a prefix of $\beta$, $\beta'_p$ is legal from $CE_{0}$.
By the definition of $\alpha_1$, $p$ does not receive any message 
from any server $q, q\neq p$ in $\beta'_p$. Moreover,  
because $m_1$ is the first message sent by some server (implicitly or explicitly),
$c_w$ receives no message from any server $q, q\neq p$ such that after the receipt of this message, $c_w$ sends a message to $p$.
By definition, $\beta'_p$ ends with a message sent by $c_w$ to $p$ (if it is not empty). It follows that 
$c_w$ does not receive any message from any server $q, q\neq p$ in $\beta'_p$.
Thus, $\beta_p$ is legal from $CE_0$. Moreover,
$RC(CE_0, \beta'_p)$ and $RC(CE_0, \beta_p)$ are indistinguishable to $p$ and $c_w$.

We next argue that $\beta_s$ is legal from $RC(CE_0, \beta_p')$. 
We notice that in $\beta_s$, only $p$ takes steps.
Then because $RC(CE_0, \beta'_p)$ and $RC(CE_0, \beta_p)$ are indistinguishable to $p$,  
by arguing that $\beta_s$ is legal from $RC(CE_0, \beta_p')$,
we also argue that $\beta_s$ is legal from $RC(CE_0, \beta_p)$.
Below we argue that $\beta_s$ is indeed legal from $RC(CE_0, \beta_p')$. 
By the definition of $\alpha_1$, 
$p$ does not receive any message from any server $q, q\neq p$ in $\beta$.
By definition of $\beta_s'$, $c_w$ does not send any message to $p$ in $\beta'_s$, 
(because all messages from $c_w$ to $p$ are sent in $\beta_p'$),
so any message $p$ receives from $c_w$ in $\beta_s$ has been sent in $\beta_p'$. 
Thus $\beta_s$ is legal from $RC(CE_0, \beta_p')$,
 and $\beta_s$ is legal from $RC(CE_0, \beta_p)$. 
Moreover, 
by arguing that $\beta_s$ is legal from $RC(CE_0, \beta_p')$,
we also show that $RC(CE_0, \beta_p'\cdot\beta_s)$ and $C_1 = RC(CE_0, \beta_p'\cdot\beta_s')$ are indistinguishable to $p$. 
Because $RC(CE_0, \beta'_p)$ and $RC(CE_0, \beta_p)$ are indistinguishable to $p$,
it follows that 
$RC(CE_0, \beta_{new})$ and $C_1$ are indistinguishable to $p$.

Now we are ready to construct $\gamma$ using 
executions 
$\sigma_{old}(CE_0, p, c_r)$,
$\sigma_{new}(C_1, p, c_r)$ 
and $\beta_{new}$ 
and referring to configuration $C_{new}(C_1, p,$ $ c_r)$.
(For simplicity, we omit 
$(CE_0, p, c_r)$
and 
$(C_1, p, c_r)$
from the notations below.)
Starting from $CE_0$, 
$\sigma_{old}$ is first applied.
Recall that for any server $q, q\neq p$, in its last step of $\sigma_{old}$, $q$ sends at most one message $mo_q'$ to $c_r$.
Since $CE_0$ and $RC(CE_0, \sigma_{old})$ are indistinguishable to $c_w$ and $p$, 
and since only $c_w$ and $p$ take steps in $\beta_{new}$, 
$\beta_{new}$ is legal from $RC(CE_0, \sigma_{old})$.
The resulting configuration is $RC(CE_0, \sigma_{old}\cdot\beta_{new})$.
Since the processes which take steps in $\sigma_{old}$ and $\beta_{new}$ are disjoint,
then $RC(CE_0, \sigma_{old} \cdot \beta_{new})$ and $RC(CE_0,  \beta_{new}\cdot \sigma_{old} )$ are indistinguishable.
Recall that to $p$, $RC(CE_0, \beta_{new})$ and $C_1$ are indistinguishable.
Therefore, to $p$, 
$RC(C_0,  \beta_{new}\cdot \sigma_{old} )$ and $C_{new}$
are indistinguishable.
Because only $p$ takes steps in $\sigma_{new}$,
$\sigma_{new}$ is therefore legal from 
$RC(CE_0,  \beta_{new}\cdot \sigma_{old} )$
and also legal from $RC(CE_0,  \sigma_{old}\cdot\beta_{new})$.
We apply it from 
$RC(CE_0,  \sigma_{old}\cdot\beta_{new})$.
Recall that in the last step of $\sigma_{new}$, $p$ sends message $mn_p'$ to $c_r$.
Finally, $mo_q'$ for any server $q, q\neq p$ (if $q$ sends a response) and $mn_p'$ are delivered in $\gamma$, and next, 
$c_r$ takes steps. Since $T_r$ is a fast transaction, $T_r$ completes. 
This concludes the construction of $\gamma$.

In $\gamma$, $c_r$ executes only transaction $T_r$.
Thus, $c_r$ decides the response for $T_r$ based solely on  the values included in $mo_q'$ for any server $q, q\neq p$ and $mn_p'$.
Since $mo_q'$ for any server $q, q\neq p$ is sent before $c_w$ takes steps, $mo_q'$ contains no value in $\{x_0, x_1, \ldots, x_N\}$.
Recall that $mn_p'$ is sent in $\gamma_{new}$.
By the one-value messages property, $mn_p'$ contains only values in $\{x_0, x_1, \ldots, x_N\}$. 
As the system is partially replicated, $mn_p'$ cannot contain all values in $\{x_0, x_1, \ldots, x_N\}$. 
It follows that in $\gamma$, the return value of $c_r$ for $T_r$ 
can be neither $(x_0, x_1, \ldots, x_N)$ nor $(x_0^{in}, x_1^{in}, \ldots, x_N^{in})$.
This contradicts Observation \ref{ob:gene-granularity}.
\end{proof}

\subsection{The Infinite Execution}

We now continue to prove that $\alpha$ is infinite. We do so by proving that there is an infinite sequence of executions 
$\alpha_1,\alpha_2, \ldots$ such that, all of them are distinct prefixes of $\alpha$. 
Let $\alpha_0$ be an empty execution.

\begin{lemma}
\label{lma:gene-induction}
For any integer $k\geq 1$, there exists an execution $\alpha_k$, legal from $CE_0$, 
in which only one transaction $T_w=(w(X_0)x_0, w(X_1)x_1, \ldots,$ $ w(X_N)x_N)$ is executed by client $c_w$.
Let $C_k$ be the configuration that results when $\alpha_k$ is applied from $CE_0$.
Then, the following hold:
\begin{enumerate}
\item\label{claim:gene-msg} $\alpha_k = \alpha_{k-1}\cdot\alpha_k'$, where in $\alpha_k'$, for at least one server $q$ in the system, at least one of the following occurs: 
\begin{itemize}
\item a message is sent by $q$ to a server other than $q$, or
\item a message is sent by $q$ to $c_w$ so that after $c_w$ receives this message, $c_w$ sends a message to a server other than $q$.
\end{itemize}
\item\label{claim:gene-vis} 
In $C_k$, $x_0, x_1, \ldots, x_N$ are not visible.
\end{enumerate}
\end{lemma}
\begin{proof}[Proof of Lemma \ref{lma:gene-induction}]
By induction on $k$.  For the base case ($k=1$), 
let $\alpha_0$ be the empty execution and $\alpha_1' = \alpha_1$. 
Then, Lemma \ref{lma:gene-base} implies that claim \ref{claim:gene-msg} holds. Moreover, Lemma \ref{lma:gene-e0} implies that claim \ref{claim:gene-vis} holds.

For the induction hypothesis, fix any integer $k\geq 2$ and assume 
that the claim holds for any integer $j$, $1\leq j\leq k-1$. 
We prove that the claim holds for $k$.\\

The proof of claim \ref{claim:gene-msg} follows similar arguments as those in the proof of Lemma \ref{lma:gene-base}.
We assume that the claim does not hold. We construct an execution $\gamma$ 
and show that $\gamma$ contradicts Observation \ref{ob:gene-granularity}.
By the induction hypothesis (claim \ref{claim:gene-vis}),
$x_0, x_1, \ldots, x_N$ are not visible in $C_{k-1}$.
By Observation \ref{ob:gene-not-visible}, 
there exists some client $c_r^{k}$ not in $\{c_w, c_0^{in}, c_1^{in}, \ldots, c_N^{in}\}$
such that if, starting from $C_{k-1}$, 
$c_r^{k}$ invokes $T_r = (r(X_0)*, r(X_1)*, \ldots, r(X_N)*)$, then $c_r^{k}$ returns 
$(x_0^{in}, x_1^{in}, \ldots, x_N^{in})$.
The construction of $\gamma$ utilizes execution $\sigma_{old}($ $C_{k-1}, p,c_r^{k})$
in a way similar that $\sigma_{old}(C_0,p,c_r)$ has been utilized in the construction of $\gamma$ in Lemma~\ref{lma:gene-base}.
Starting from $C_{k-1}$, $c_w$ and all servers take steps until $x_0, x_1, \ldots, x_N$ are visible.
We then define configuration $C_v$, and executions $\beta$ and $\beta_{new}$ 
in a way similar to that defined in the proof of Lemma~\ref{lma:gene-base} (but with $C_{k-1}$ playing the role
of $CE_0$). 
We finally define $\gamma$ based on $\sigma_{old}(C_{k-1}, p,c_r^{k})$, $\sigma_{new}(C_v,p,c_r^{k})$
and $\beta_{new}$, and we argue that in $\gamma$, $c_r^{k}$ returns a response for $T_r$
that contradicts Observation~\ref{ob:gene-granularity}.    

We continue with the details of the proof. 
To derive a contradiction, assume that we let $c_w$ and all servers  take steps starting from $C_{k-1}$ and the following holds:
for every server $q$,
\begin{itemize}
\item $q$ sends no message to a server other than $q$;
\item $q$ sends no message to $c_w$ so that after $c_w$ receives the message, $c_w$ sends a message to a server other than $q$.
\end{itemize}
In other words, $\alpha_k'$ cannot be defined.

To derive a contradiction, we will construct an execution $\gamma$ where in addition to $T_w$, a read-only transaction $T_r=(r(X_0)*, r(X_1)*, \ldots, $ $ r(X_N)*)$ is invoked by a client $c_r^{k}$ 
not in $\{c_w, c_0^{in}, c_1^{in}, \ldots, c_N^{in}\}$
and $c_r^{k}$ returns from $T_r$. We will show that $\gamma$ contradicts Observation \ref{ob:gene-granularity}.

To construct $\gamma$, we first define an execution $\beta$ and some subsequences of it,
and we study their properties. 
Starting from $C_{k-1}$, $c_w$ invokes $T_w$ and executes solo
(i.e., only $c_w$ and servers take steps) until $x_0, x_1, \ldots, x_N$ are visible 
(minimal progress for write-only transactions implies that there is a configuration 
in which $x_0, x_1,\ldots, x_N$ are indeed visible).
Let $C_v$ be the first configuration where $x_0, x_1, \ldots, x_N$ are visible.
Let $\beta$ be the sequence of steps which are taken from $C_{k-1}$ to $C_v$. 
%Consider execution $\gamma_{new}(C_v, p, c_r^{k})$.
%-----------------------------------------------
For any server $p$, consider execution $\gamma_{new}(C_v, p, c_r^{k})$.
Because in $\gamma_{new}(C_v, p, c_r^{k})$, starting from $C_v$, only the messages sent by $c_w$  are delivered before $c_r^{k}$ takes the steps to complete $T_r$, thus to any server $p$, $\gamma_{new}(C_v, p, c_r^{k})$ and $\gamma_{new}(C_v, s, c_r^{k})$ for any server $s, s\neq p$ are indistinguishable.
Thus when $c_r^{k}$ takes the steps to complete $T_r$, $c_r^{k}$ receives the same set of responses in $\gamma_{new}(C_v, p, c_r^{k})$ for any server $p$.
In order to exclude empty and garbage responses, let $M$ be the smallest set of non-empty responses where each response contains a value that has been written by some write transaction. As the system is partially replicated, the size of $M$ is at least two.
W.l.o.g., we fix any server $p$ which sends a response included in $M$ and thus fix execution $\gamma_{new}(C_v, p, c_r^{k})$.
Then, Observation~\ref{obs:gene-gnew} implies that the response for $T_r$ in $\gamma_{new}(C_v, p, c_r^{k})$
is $(x_0, x_1, \ldots, x_N)$. 
For simplicity, we omit 
$(C_v, p, c_r^{k})$
from the notation in the rest of the proof of this lemma.

Let $\beta'_p$ be the shortest prefix of $\beta$ which contains all messages sent by $c_w$ to $p$.
Let $\beta_p$ be the subsequence of $\beta'_p$ in which all steps taken by any server $q, q\neq p$ have been removed.
Let $\beta'_s$ be the suffix of $\beta$ after
$\beta_p'$.
Let $\beta_s$ be the subsequence of $\beta'_s$ containing only steps by $p$.
Let $\beta_{new}$ be $\beta_p\cdot \beta_s$.
We now argue that $\beta_{new}$ is legal from $C_{k-1}$.
Because $\beta'_p$ is a prefix of $\beta$, $\beta'_p$ is legal from $C_{k-1}$.
By assumption, if $p$ receives any message 
from any server $q, q\neq p$ in $\beta'_p$, the message has been sent in $\alpha_{k-1}$, i.e., before the prefix $\beta'_p$.
Moreover,  
by assumption, after $\alpha_{k-1}$,
any server $q, q\neq p$ sends no message to $c_w$ so that after $c_w$ receives the message, $c_w$ sends a message to $p$.
By definition, $\beta'_p$ ends with a message sent by $c_w$ to $p$ (if $\beta'_p$ is not empty). It follows that 
if $c_w$ receives any message from any server $q, q\neq p$ in $\beta'_p$, then the message has been sent in $\alpha_{k-1}$, i.e., before the prefix $\beta'_p$.
Thus, $\beta_p$ is legal from $C_{k-1}$. Moreover,
$RC(C_{k-1}, \beta'_p)$ and $RC(C_{k-1}, \beta_p)$ are indistinguishable to $p$ and $c_w$.

We next argue that $\beta_s$ is legal from $RC(C_{k-1}, \beta_p')$. 
We notice that in $\beta_s$, only $p$ takes steps.
Then because $RC(C_{k-1}, \beta'_p)$ and $RC(C_{k-1}, \beta_p)$ are indistinguishable to $p$,  
by arguing that $\beta_s$ is legal from $RC(C_{k-1}, \beta_p')$,
we also argue that $\beta_s$ is legal from $RC(C_{k-1}, \beta_p)$.
Below we argue that $\beta_s$ is indeed legal from $RC(C_{k-1}, \beta_p')$. 
By assumption, 
if $p$ receives any message from any server $q, q\neq p$ in $\beta_s$,
then the message has been sent in $\alpha_{k-1}$, i.e., before $\beta$.
By definition of $\beta_s'$, $c_w$ does not send any message to $p$ in $\beta'_s$, 
(because all messages from $c_w$ to $p$ are sent in or before $\beta_p'$),
so any message $p$ receives from $c_w$ in $\beta_s$ has been sent in or before $\beta_p'$. 
Thus $\beta_s$ is legal from $RC(C_{k-1}, \beta_p')$,
 and $\beta_s$ is legal from $RC(C_{k-1}, \beta_p)$. 
Moreover, 
by arguing that $\beta_s$ is legal from $RC(C_{k-1}, \beta_p')$,
we also show that $RC(C_{k-1}, \beta_p'\cdot\beta_s)$ and $C_v = RC(C_{k-1}, \beta_p'\cdot\beta_s')$ are indistinguishable to $p$. 
Because $RC(C_{k-1}, \beta'_p)$ and $RC(C_{k-1}, \beta_p)$ are indistinguishable to $p$,
it follows that 
$RC(C_{k-1}, \beta_{new})$ and $C_v$ are indistinguishable to $p$.

To construct $\gamma$, we have to utilize executions $\sigma_{old}(C_{k-1}, p, c_r^{k})$
and $\sigma_{new}(C_v, p, c_r^{k})$%. 
 and refer to configuration $C_{new}(C_v, p, c_r^{k})$.
(For simplicity, we omit $(C_{k-1}, p, c_r^{k})$
and $(C_v, p, c_r^{k})$ from the notations below.)
Specifically, starting from $C_{k-1}$, 
$\sigma_{old}$ is first applied.
Recall that for any server $q, q\neq p$, in its last step of $\sigma_{old}$, $q$ sends at most one message $mo_q'$ to $c_r^{k}$.
Since $C_{k-1}$ and $RC(C_{k-1}, \sigma_{old})$ are indistinguishable to $c_w$ and $p$, and since only $c_w$ and $p$ 
take steps in $\beta_{new}$, $\beta_{new}$ is legal from $RC(C_{k-1}, \sigma_{old})$. So, to construct $\gamma$,
$\beta_{new}$ is applied from $RC(C_{k-1}, \sigma_{old})$.
The resulting configuration is $RC(C_{k-1}, \sigma_{old} \cdot \beta_{new})$.
Since the processes which take steps in $\sigma_{old}$ and $\beta_{new}$ are disjoint,
then $RC(C_{k-1}, \sigma_{old} \cdot \beta_{new})$ and $RC(C_{k-1},  \beta_{new}\cdot \sigma_{old} )$ are indistinguishable.
Recall that to $p$, $RC(C_{k-1}, \beta_{new})$ and $C_v$ are indistinguishable.
Therefore, to $p$,
$RC(C_{k-1},  \beta_{new}\cdot \sigma_{old} )$ and $C_{new}$
are indistinguishable.
Because only $p$ takes steps in $\sigma_{new}$,
$\sigma_{new}$ is therefore legal from 
$RC(C_{k-1},  \beta_{new}\cdot \sigma_{old} )$
and also legal from $RC(C_{k-1},  \sigma_{old}\cdot\beta_{new})$.
We apply it from 
$RC(C_{k-1},  \sigma_{old}\cdot\beta_{new})$.
Recall that in the last step of $\sigma_{new}$, $p$ sends a message $mn_p'$ to $c_r^{k}$.
Finally, $mo_q'$ for any server $q, q\neq p$ (if $q$ sends a response) and $mn_p'$ are delivered in $\gamma$, and next,  
$c_r^{k}$ takes steps. Since $T_r$ is a fast transaction, $T_r$ completes. 
This concludes the construction of $\gamma$. 

In $\gamma$, $c_r^{k}$ executes only transaction $T_r$.
Thus, $c_r^{k}$ decides the response for $T_r$ based solely on  the values included in $mo_q'$ for any server $q, q\neq p$ and $mn_p'$.
Since $mo_q'$ for any server $q, q\neq p$ is sent in $\gamma_{old}(C_{k-1}, p,c_r^{k})$, 
by Observation \ref{obs:gene-gold} and the one-value messages property, 
$mo_q'$ for any server $q, q\neq p$ contains no value in $\{x_0, x_1, \ldots, x_N\}$. 
Recall that $mn_p'$ is sent in $\gamma_{new}$.
By Observation \ref{obs:gene-gnew} and the one-value messages property, $mn_p'$ contains values in $\{x_0, x_1, \ldots, x_N\}$. 
It follows that in $\gamma$, the return value of $c_r^{k}$ for $T_r$ 
can be neither $(x_0^{in}, x_1^{in}, \ldots, x_N^{in})$ nor
$(x_0, x_1, \ldots, x_N)$.
This contradicts Observation \ref{ob:gene-granularity}.\\

To prove claim~\ref{claim:gene-vis}, we use similar arguments as those in the proof of Lemma~\ref{lma:gene-e0}
with configuration $C_{k-1}$ playing the role of $CE_0$ and configuration $C_k$ playing the role of $C_1$.
We assume that the claim does not hold. We construct an execution $\delta$ 
and show that $\delta$ contradicts Observation \ref{ob:gene-granularity}.
The construction of $\delta$ utilizes the same execution $\sigma_{old}(C_{k-1}, p,c_r^{k})$ as in the proof of claim \ref{claim:gene-msg}.
Starting from $C_{k-1}$, $c_w$ and all servers take steps as defined in $\alpha_k'$.
We then define executions $\rho$ and $\rho_{new}$ 
in a way similar to $\beta$ and $\beta_{new}$  defined in the proof of Lemma~\ref{lma:gene-e0} (but with $C_{k-1}$ playing the role
of $CE_0$ and $C_{k}$ playing the role
of $C_1$). 
We finally define $\delta$ based on $\sigma_{old}(C_{k-1}, p, c_r^{k})$, $\sigma_{new}(C_k,p,c_r^{k})$
and $\rho_{new}$, and we argue that in $\delta$, $c_r^{k}$ returns a response for $T_r$
that contradicts Observation~\ref{ob:gene-granularity}.    

We now present the details of the proof. 
Assume that in $C_k$, $x_i$ is visible for some $i\in\{0,1, \ldots, N\}$.
We will derive a contradiction by constructing an execution $\delta$ where in addition to $T_w$, a read-only transaction $T_r=(r(X_0)*, r(X_1)*, \ldots, r(X_N)*$ 
is invoked by a client $c_r^{k}$ 
not in $\{c_w, c_0^{in}, c_1^{in}, \ldots, c_N^{in}\}$
and $c_r^{k}$ returns from $T_r$. We will show that $\delta$ contradicts Observation \ref{ob:gene-granularity}.

We construct $\delta$ by utilizing $\gamma_{old}(C_{k-1}, p, c_r^{k})$
and $\gamma_{new}(C_k, p, c_r^{k})$. 
By Observation~\ref{obs:gene-gold}, $c_r^{k}$ returns $(x_0^{in}, x_1^{in}, \ldots, x_N^{in})$
for $T_r$ in $\gamma_{old}(C_{k-1}, $ $p, c_r^{k})$.
Because by assumption, $x_i$ is visible at $C_k$, Observation~\ref{obs:gene-gnew} implies that 
$c_r^{k}$ returns $(x_0, x_1, \ldots, x_N)$ for $T_r$ in $\gamma_{new}(C_k, p_{k\%2}, c_r^{k})$.

To construct $\delta$, we first define an execution $\rho$ and some subsequences of it,
and we study their properties.
Let $\rho$ be the sequence of steps which are taken from $C_{k-1}$ to $C_k$.
Let $\rho'_p$ be the shortest prefix of $\rho$ which contains all messages sent by $c_w$ to $p$.
Let $\rho_p$ be the subsequence of $\rho'_p$ in which all steps taken by any server $q, q\neq p$ have been removed.
Let $\rho'_s$ be the suffix of $\rho$ after $\rho_p'$.
Let $\rho_s$ be the subsequence of $\rho'_s$ containing only steps by $p$.
Let $\rho_{new}$ be $\rho_p\cdot \rho_s$.
We now argue that $\rho_{new}$ is legal from $C_{k-1}$.
Because $\rho'_p$ is a prefix of $\rho$, $\rho'_p$ is legal from $C_{k-1}$.
By the definition of $\alpha_k$, if $p$ receives any message 
from any server $q, q \neq p$ in $\beta'_p$, the message has been sent in $\alpha_{k-1}$, i.e., before the prefix $\beta'_p$.
Moreover,  
because $\alpha_k'$ ends with the first message sent by some server to another server (implicitly or explicitly),
$c_w$ receives no message from any server $q, q\neq p$ such that after the receipt of this message, $c_w$ sends a message to $p$.
By definition, $\rho'_p$ ends with a message sent by $c_w$ to $p$ (if $\rho'_p$ is not empty). It follows that 
if $c_w$ receives any message from any server $q, q\neq p$ in $\rho'_p$, then the message has been sent in $\alpha_{k-1}$, i.e., before the prefix $\rho'_p$.
Thus, $\rho_p$ is legal from $C_{k-1}$. Moreover,
$RC(C_{k-1}, \rho'_p)$ and $RC(C_{k-1}, \rho_p)$ are indistinguishable to $p$ and $c_w$.

We next argue that $\rho_s$ is legal from $RC(C_{k-1}, \rho_p')$. 
We notice that in $\rho_s$, only $p$ takes steps.
Then because $RC(C_{k-1}, \rho'_p)$ and $RC(C_{k-1}, \rho_p)$ are indistinguishable to $p$,  
by arguing that $\rho_s$ is legal from $RC(C_{k-1}, \rho_p')$,
we also argue that $\rho_s$ is legal from $RC(C_{k-1}, \rho_p)$.
Below we argue that $\rho_s$ is indeed legal from $RC(C_{k-1}, \rho_p')$. 
By the definition of $\alpha_k$, 
if $p$ receives any message from any server $q, q\neq p$ in $\rho_s$,
then the message has been sent in $\alpha_{k-1}$, i.e., before $\rho$.
By definition of $\rho_s'$, 
$c_w$ does not send any message to $p$ in $\rho'_s$, 
(because all messages from $c_w$ to $p$ are sent in or before $\rho_p'$),
so any message $p$ receives from $c_w$ in $\rho_s$ has been sent in or before $\rho_p'$. 
Thus $\rho_s$ is legal from $RC(C_{k-1}, \rho_p')$,
 and $\rho_s$ is legal from $RC(C_{k-1}, \rho_p)$. 
Moreover, 
by arguing that $\rho_s$ is legal from $RC(C_{k-1}, \rho_p')$,
we also show that $RC(C_{k-1}, \rho_p'\cdot\rho_s)$ and $C_k = RC(C_{k-1}, \rho_p'\cdot\rho_s')$ are indistinguishable to $p$. 
Because $RC(C_{k-1}, \rho'_p)$ and $RC(C_{k-1}, \rho_p)$ are indistinguishable to $p$,
it follows that 
$RC(C_{k-1}, \rho_{new})$ and $C_k$ are indistinguishable to $p$.

Now we are ready to construct $\delta$ using 
executions 
$\sigma_{old}(C_{k-1}, p, c_r^{k})$,
$\sigma_{new}(C_k, p, c_r^{k})$ 
and $\rho_{new}$,
referring to configuration $C_{new}(C_k, p, c_r^{k})$. 
(For simplicity, we omit $(C_{k-1}, p, c_r^{k})$
and $(C_k, p, c_r^{k})$
from the notations below, which abuses $\sigma_{new}$ and $C_{new}$.)
Starting from $C_{k-1}$, 
$\sigma_{old}$ is first applied.
Recall that in its last step of $\sigma_{old}$, any server $q, q\neq p$ sends at most one message $mo_q'$ to $c_r^{k}$.
Since $C_{k-1}$ and $RC(C_{k-1}, \sigma_{old})$ are indistinguishable to $c_w$ and $p$, 
and since only $c_w$ and $p$ take steps in $\rho_{new}$, 
$\rho_{new}$ is legal from $RC(C_{k-1}, \sigma_{old})$.
The resulting configuration is $RC(C_{k-1}, \sigma_{old}\cdot\rho_{new})$.
Since the processes which take steps in $\sigma_{old}$ and $\rho_{new}$ are disjoint,
then $RC(C_{k-1}, \sigma_{old} \cdot \rho_{new})$ and $RC(C_{k-1},  \rho_{new}\cdot \sigma_{old} )$ are indistinguishable.
Recall that to $p$, $RC(C_{k-1}, \rho_{new})$ and $C_k$ are indistinguishable.
Therefore, to $p$, 
$RC(C_{k-1},  \rho_{new}\cdot \sigma_{old} )$ and $C_{new}$
are indistinguishable.
Because only $p$ takes steps in $\sigma_{new}$,
$\sigma_{new}$ is therefore legal from 
$RC(C_{k-1},  \rho_{new}\cdot \sigma_{old} )$
and also legal from $RC(C_{k-1},  \sigma_{old}\cdot\rho_{new})$.
We apply it from 
$RC(C_{k-1},  \sigma_{old}\cdot\rho_{new})$.
Recall that in the last step of $\sigma_{new}$, $p$ sends message $mn_p'$ to $c_r^{k}$.
Finally, $mo_q'$ for any server $q, q\neq p$ (if $q$ sends a response) and $mn_p'$ are delivered in $\delta$, and next, 
$c_r^{k}$ takes steps. Since $T_r$ is a fast transaction, $T_r$ completes. 
This concludes the construction of $\delta$.

In $\delta$, $c_r^{k}$ executes only transaction $T_r$.
Thus, $c_r^{k}$ decides the response for $T_r$ based solely on  the values included in $mo_q'$ for any server $q, q\neq p$ and $mn_p'$.
Since $mo_q'$ for any server $q, q\neq p$ is sent in $\gamma_{old}(C_{k-1}, p,c_r^{k})$, 
by Observation \ref{obs:gene-gold} and the one-value messages property, 
$mo_q'$ for any server $q, q\neq p$ contains no value in $\{x_0, x_1, \ldots, x_N\}$.
Recall that $mn_p'$ is sent in $\delta_{new}$.
By Observation \ref{obs:gene-gnew} and the one-value messages property, $mn_p'$ contains only values in $\{x_0, x_1, \ldots, x_N\}$.
It follows that in $\delta$, the return value of $c_r^{k}$ for $T_r$ 
can be neither $(x_0^{in}, x_1^{in}, \ldots, x_N^{in})$ 
nor $(x_0, x_1, \ldots, x_N)$.
This contradicts Observation \ref{ob:gene-granularity}.
It follows that claim \ref{claim:gene-vis} is true for $k$.
\end{proof}

\end{document}